%% file: main.tex
\documentclass[11pt]{article}

\usepackage{macros}

\usepackage{thmtools}
\usepackage{hyperref}

\usepackage{algorithm, algorithmic}
\newcommand{\defeq}{\vcentcolon=}
\newcommand{\smax}{\mathsf{smax}}
\newcommand{\oracle}{\mathsf{oracle}}
\newcommand{\hc}{\{\pm 1\}^n}
\newcommand*\diff{\mathop{}\!\mathrm{d}}

\newenvironment{nohyphens}{%
  \par
  \hyphenpenalty=10000
  \exhyphenpenalty=10000
  \sloppy % Makes TeX obey margins by stretching inter word spaces
}{\par}

\title{Discrepancy Minimization via Regularization}

\author{Lucas Pesenti\footnote{Bocconi University, \texttt{lucas.pesenti@phd.unibocconi.it}} \and Adrian Vladu\footnote{CNRS \& IRIF, Université Paris Cité, \texttt{vladu@irif.fr}}}
\date{}
\begin{document}

\maketitle

\begin{abstract}
\begin{nohyphens}

    We introduce a new algorithmic framework for discrepancy minimization based on \textit{regularization}. We demonstrate how varying the regularizer allows us to re-interpret several breakthrough works in algorithmic discrepancy, ranging from Spencer's theorem \cite{S85, B10} to Banaszczyk's bounds \cite{B98, BDG19}. Using our techniques, we also show that the Beck-Fiala and \Komlos conjectures are true in a new regime of \textit{pseudorandom} instances.
\end{nohyphens}
\end{abstract}

\input{intro}
\input{overview}

\input{approach}
\input{spencer}
\input{pseudorandom}
\input{consequences}

\bibliographystyle{alpha}
\bibliography{main}

\pagebreak

\input{ellipsoid}

\end{document}

%% file: intro.tex
\section{Introduction}

Discrepancy theory is a subfield of combinatorics which has branched in computer science due to its several connections to geometric problems, randomized algorithms, and complexity theory~\cite{M09,chazelle2000discrepancy}.

A landmark result in the field is Spencer's celebrated “Six standard deviations suffice”~\cite{S85}. In its simplest form, Spencer's paper considers a set system $S$ of cardinality $n$ over a ground set of $n$ elements. The problem is to color each element of the ground set in red or blue, in such a way that all the sets are balanced, namely do not contain many more red than blue elements or vice versa. The maximum imbalance of a set induced by a coloring is called the \textit{discrepancy} of the coloring.

Spencer's result offers an important insight into the limitations of the tools we generally employ to prove the existence of mathematical objects. A standard method to show that there always exists a low discrepancy coloring is to prove that a random coloring produces one with nonzero probability. To this end, one shows that for each set in $S$, a random coloring will have small discrepancy with high probability. Applying a union bound turns this into a simultaneous guarantee for all sets in $S$. This standard idea shows that one can always produce a coloring of discrepancy $O\left(\sqrt{n\log n}\right)$. Spencer~\cite{S85} shows that this approach misses even better colorings. Using a difficult nonconstructive argument, he proves that in fact, colorings of discrepancy $6\sqrt{n}$ exist. This result is tight up to constant factors, and exhibits an example where correlations between different sets in $S$ can be exploited in order to overcome the limitations of the union bound technique, which essentially tries to ignore them.

Although in~\cite{S85} it is conjectured that there are inherent limitations to finding constructively a low discrepancy coloring, Bansal~\cite{B10} proved the contrary by exhibiting a polynomial-time algorithm whose output matches Spencer's bound up to constant factors. This provided a new direction for attacking open problems in the field, and was followed by a deluge of new algorithmic results~\cite{LM15,ES18,R17,LRR17,BDG19,BDGL19,BM20,BLV22}. Notably, some of these results gave efficient algorithms to construct colorings for structured set systems, aiming towards breaking longstanding upper bounds in discrepancy theory, which are captured by the well-known Beck-Fiala and \Komlos conjectures.

Interestingly, the recent works on discrepancy have increasingly relied on tools from continuous optimization, which have already shown great results in the area of fast algorithms. Hence, Bansal~\cite{B10} constructs a coloring by maintaining a fractional solution which gets updated by iteratively solving a semidefinite program, Lovett and Meka~\cite{LM15} perform a random walk through a polytope, Eldan and Singh~\cite{ES18} solve a linear program, and Levy, Ramadas and Rothvoss~\cite{LRR17} use an algorithm inspired by the multiplicative weights update method. Given the excellent results already delivered by continuous methods in other areas, using them to improve longstanding bounds in discrepancy theory remains an intriguing research direction. To this extent, it is tempting to ask whether casting these questions in an appropriate optimization framework will provide a natural approach that will lead to the sought answers.

\subsection{Our contribution}

In this paper, we further extend the connections between discrepancy and continuous optimization by providing a simple framework for discrepancy minimization based on a series of invocations of Newton's method on a regularized objective. Using this basic framework, we provide a twofold contribution.

\begin{itemize}
    \item We show that using a version of the regularizer from~\cite{ALO15}, we obtain a simple and elegant proof of Spencer's result via a constructive algorithm. Using slightly different parameters, we also argue that in fact, ``$4.1$ standard deviations suffice".

    As previously noted in related works~\cite{ALO15}, the choice of the regularizer may prove critical to achieving the correct bound. Indeed, while~\cite{LRR17} regularize their problem with a negative entropy term, so that their updates reflect those corresponding to the multiplicative weights update method, we use instead the regularizer from~\cite{ALO15}. This allows for a simpler and tighter way to control the increase of the (regularized) discrepancy over the course of the algorithm. This method parallels previous developments~\cite{LRR17,BLV22}, while simultaneously providing a clean regularization framework for the problem.
    
    \item We extend the ideas of Potukuchi~\cite{P20}, who gave an improved bound for the Beck-Fiala problem in the case where the input matrix is ``pseudorandom''.
    More precisely, we show that for \Komlos instances, given an $n\times n$ matrix $A$ with columns of at most unit $\ell_2$-norm, we can achieve a discrepancy of $O(1+\sqrt{\lambda \log n})$, where $\lambda =  \max_{v \perp \bf{1}} \|A^{\odot 2} v\|_2 / \|v\|_2$ (and $\odot$ denotes entry-wise product).

    This automatically improves  Potukuchi's result from $O(\sqrt{s}+\lambda)$ to $O(\sqrt{s} + \sqrt{\lambda \log n})$ on Beck-Fiala instances with column-sparsity $s$, in the regime where $\lambda = \Omega(\log n)$. In addition, it implies that the \Komlos conjecture is true for random rotation matrices and random Gaussian matrices.
\end{itemize}

Finally, we believe that this framework is powerful enough to provide new paths to attacking the major conjectures in the area.

\subsection{Related work}

Following Spencer's original paper~\cite{S85}, an exciting series of algorithmic results has emerged. The first to provide a constructive proof was Bansal~\cite{B10}, following which renewed efforts have attempted to make progress on the Beck-Fiala and \Komlos conjectures.

In particular, Lovett and Meka~\cite{LM15} provide another algorithm recovering Spencer's theorem, based on a random walk in a polytope. Rothvoss~\cite{R17} shows another approach based on projections, which crucially relies on convex geometric arguments involving Gaussian width. Eldan and Singh~\cite{ES18} recover the same result by solving a linear program with a random linear objective. Levy, Ramadas and Rothvoss~\cite{LRR17} give another algorithm inspired by the multiplicative weights update method. Very recently, Bansal, Laddha and Vempala~\cite{BLV22} provided an algorithm using $\ell_p$-norms as potential functions, which is very similar to our approach.

For the Beck-Fiala and \Komlos problems, we note the algorithms from Bansal, Dadush and Garg~\cite{BDG19} and Bansal, Dadush, Garg and Lovett~\cite{BDGL19}, that match Banaszczyk's bound of $O(\sqrt{\log n})$ for the \Komlos problem, which  in turn automatically implies a $O(\sqrt{s \log n})$ discrepancy bound for Beck-Fiala. In addition, a lot of work has focused on studying random instances of these conjectures \cite{EL19, BM20, HR19, P20, AN22}. In particular, Potukuchi \cite{P20} extracted a pseudorandom property of instances that is sufficient to obtain a low-discrepancy coloring.

We also acknowledge a series of very exciting recent developments on extended versions of Spencer's problem~\cite{reis2020vector} and generalizations of it based on mirror descent and communication complexity~\cite{hopkins2022matrix, dadush2022new}.

Concerning regularization theory, while regularization techniques related to our $\ell_q$-regularizer have long existed in the bandit literature~\cite{bubeck2012regret}, the first time such an idea was successfully employed in a combinatorial context was for computing optimal spectral sparsifiers~\cite{BSS14}, where the authors used a barrier argument to track the evolution of the eigenvalues of a matrix. Later, Allen-Zhu, Liao and Orecchia~\cite{ALO15} further developed this technique and made the connection to regularization explicit. This allowed them to improve the running time of the spectral sparsification algorithms from~\cite{BSS14}. 

We further comment on the link of some of these previous works with our approach in \pref{sec:comparison}.

\subsection{Organization of the paper}

We will formally introduce our algorithmic framework in \pref{sec:approach}. \pref{sec:generic} is devoted to reviewing the underlying iterative algorithm that is shared with several previous works. In \pref{sec:potential}, we will introduce different ways to regularize the discrepancy objective, and we will prove the corresponding technical bounds in \pref{sec:bounds}.

Our first application of these tools is to the setting of Spencer's theorem. We will give three different proofs of Spencer's theorem, all based on the same idea but with different goals in mind:
\begin{itemize}
    \item In \pref{sec:overview}, we sketch the main ideas of our framework on square matrices using Newton steps in the continuous limit.
    \item In \pref{sec:spencereasy}, we discretize the previous approach and give a full algorithmic proof of Spencer's theorem in the general case.
    \item In \pref{sec:spencerconstant}, we come back to the square setting and give a more careful analysis of the algorithm in order to optimize the leading constant. We improve Spencer's original constant from $5.4$ to $4.1$.
\end{itemize}

\pref{sec:main} is dedicated to the proof of our new bounds for pseudorandom instances of \Komlos and Beck-Fiala conjectures. Our main results are \pref{thm:prkomlos} and \pref{thm:prbf}, which we prove in \pref{sec:strategy} to \pref{sec:together}. We present the consequences for random instances in \pref{sec:corollary}.

Finally, we discuss some future research directions in \pref{sec:future}.

\subsection{Notations}

Let $n$ and $m$ be some positive integers.

We will denote by $\log$ the natural logarithm. We set $[n]\defeq \{1, \ldots, n\}$.  We define the $(n-1)$-dimensional simplex by $\Delta_n\defeq \{r\in\R_+^n : \sum_{i\le n} r_i=1\}$, which is the set of probability distributions supported on $[n]$. We will use $\bf{1}$ as a shortcut for the vector $(1, \ldots, 1)$.

We consider the usual Euclidean inner product $\langle x,y\rangle\defeq \sum_{i\le n} x_i y_i$ on vectors $x,y\in\R^n$, and $\langle A,B\rangle\defeq \sum_{1\le i\le m, 1\le j\le n} A_{ij} B_{ij}$ on matrices $A,B\in\R^{m\times n}$. The $\ell_2$-norm of a vector will by denoted by $\| x\|_2\defeq \langle x,x\rangle^{1/2}$ and its $\ell_\infty$-norm by $\lVert x\rVert_\infty\defeq \max_{i\le n} |x_i|$. The spectral norm (or $\ell_2\to\ell_2$ operator norm) of a matrix $A\in\R^{m\times n}$ will be denoted by $\| A\|_{\text{op}}\defeq \max_{\lVert x\rVert_2=1} \lVert Ax\rVert_2$ and its Frobenius norm by $\lVert A\rVert_F\defeq\langle A,A\rangle^{1/2}$. If $A\in\R^{m\times n}$ and $i\in [m]$, we write $A_i$ for the $i$-th column of the matrix $A^\top$ (namely, the $i$-th row of $A$ treated as a column vector). Similarly, if $j\in [n]$, $A^j$ will denote the $j$-th column of $A$. We use the Hadamard notation $\odot$ to denote the entrywise product of vectors: for any $x,y\in\R^n$, we let $(x\odot y)_i\defeq x_i y_i$ for all $i\in [n]$. We also set $x^{\odot 2}\defeq x\odot x$. Given $A,B\in\R^{m\times n}$, we similarly define $(A\odot B)_{ij}\defeq A_{ij} B_{ij}$ for all $i\in [m], j\in [n]$.  Given a vector $x\in\R^n$, we define $\diag(x)$ to be the $n\times n$ diagonal matrix with the elements of $x$ on the diagonal.  Given $F\subseteq [n]$ and $x\in \R^F$, we will view interchangeably $x$ as a vector on $\R^F$ and on $\R^n$, where the coordinates in $[n]\setminus F$ are filled with zeros.

Unless specified otherwise, the $\lesssim$, $O$ and $\Omega$ notations will hide only universal constants.

%% file: overview.tex
\section{Technical overview: Spencer's theorem via Newton steps in the continuous limit}
\label{sec:overview}

In this section, we give an informal proof of a result of Spencer that illustrates how regularization comes into play in discrepancy minimization. For now, we aim at keeping the discussion simple and defer most details to \pref{sec:approach} and \pref{sec:spencer}.

\begin{theorem}[\cite{S85}]
	\label{thm:spencer}
	There exists a universal constant $K>0$ such that for any matrix $A\in[-1,1]^{n\times n}$, there exists $x\in\{\pm 1\}^n$ such that
	\[
	    \lVert Ax\rVert_\infty \le K\sqrt{n}\mper
	\]
\end{theorem}

Many different proofs of \pref{thm:spencer} are already known: purely combinatorial \cite{S85}, with insights from convex geometry \cite{G89, G97, R17, ES18}, via random walks \cite{LM15}, with the multiplicative weights update method \cite{LRR17}, or via barrier potential functions \cite{BLV22}. We will compare our approach with these last three techniques in \pref{sec:comparison}.

\paragraph{Sticky walk.} We build a deterministic sequence $x(t)\defeq (x_1(t),\ldots,x_n(t))$ for times $t\in [0,T]$. At any time step $t$, $x(t)$ will be an element of the solid hypercube $[-1,1]^n$ that represents a partial coloring.  We start from $x(0)\defeq (0,\ldots,0)$ and the dynamic ends when $x(t)$ hits a corner of the hypercube. We define the set of active coordinates of a fractional coloring $x\in [-1,1]^n$ as
\[
    F\defeq \{j\in [n] : x_j \notin \{-1,1\}\}\mper
\]
The final algorithm described in \pref{sec:approach} and \pref{sec:spencer} will essentially be a discretization of this continuous dynamic.

\paragraph{Regularized maximum.} In order to control the quantity $\|Ax\|_\infty$ over the duration of the walk, we now define a smooth proxy for the $\infty$-norm. Since we can always add the negation of all the rows to the matrix $A$, without loss of generality it suffices to track $\max_{i\in [n]} (Ax)_i$. Naturally, for any $y\in\R^n$,
\[
    \max_{i\in [n]} y_i =\max_{r\in\Delta_n} \langle r,y\rangle,\quad\text{where $\Delta_n\defeq \left\{r\in \R_+^n:\sum_{i\le n} r_i =1\right\}$}\mper
\]
Instead, we consider the following regularized version of the right-hand side, which is the maximization problem where we added an $\ell_{1/2}$-type penalty for each element of the simplex:
\[
    \omega^*(y)\defeq \max_{r\in\Delta_n} \langle r,y\rangle + \sum_{i\le n} r_i^{\frac{1}{2}}\mper
\]
In what follows, $\omega^*(Ax)$ will play a role of proxy for the $\infty$-norm. It is not hard to see (\pref{lem:explain}) that we only lose a $\sqrt{n}$ additive factor through this approximation. Therefore, for proving \pref{thm:spencer}, it suffices to bound the total increase of the regularized maximum. We will discuss in \pref{sec:spencereasy} the choice of this particular $\ell_{1/2}$-regularizer.

\begin{algorithm}
	\caption{Continuous dynamic for discrepancy minimization}
	\label{alg:iterativecont}
	\begin{algorithmic}[1]
		\WHILE{$F\neq \varnothing$}
            \STATE $\dot{x} =\underset{\delta : \langle \delta, x\rangle=0 \text{ and } \textnormal{supp}(\delta) \subseteq F}{\argmin}\, \langle A^{\top}\nabla \omega^*(Ax), \delta \rangle + \frac{1}{2}  \cdot \delta^\top A^\top \nabla^2\omega^*(Ax)  A \delta$
            \STATE $F=\{i\in [n] : x_i\notin \{-1,1\}\}$
		\ENDWHILE
		\RETURN$x$.
	\end{algorithmic}
\end{algorithm}

\paragraph{Continuous dynamic} We sketch our dynamic in pseudo-code in \pref{alg:iterativecont}. Essentially, we impose two conditions on the update direction $\delta$: $\text{supp}(\delta)\subseteq F$ ensures that the walk stays in the solid hypercube by fixing the coordinates of $x$ when they reach $\pm 1$, while $\langle \delta,x\rangle= 0$ ensures that the dynamic will eventually converge to a corner of the hypercube. Under these constraints, we select the direction that minimizes the best quadratic approximation of our potential function $\omega^*(Ax)$. In this sense, this is essentially a Newton step.

%Note that \pref{eq:evol} might be ill-defined when $|F(t)|\le 1$. Therefore, we choose to stop the dynamic when $|F(t)|\le 1$ and observe that in this case, $x(T)$ can be rounded into an element of $\{-1,1\}^n$ by only losing discrepancy $1$, which is negligible in our setting.

\paragraph{Local analysis.} We would like to bound the increase in potential,
\[
    \frac{\diff \omega^*(Ax)}{\diff \|x\|_2^2} \approx \langle \nabla \omega^*(Ax),A\delta\rangle+\frac{1}{2}\langle A\delta ,\nabla^2 \omega^*(Ax)A\delta\rangle\mcom
\]
where $\delta$ is the minimizer on line 2 of \pref{alg:iterativecont}.

Since the quadratic term is invariant by sign changes $\pm \delta$, we can always upper bound the term that is linear in $\delta$ by 0. Thus, it suffices to prove that the matrix $A^\top\nabla^2\omega^*(Ax)A$ has a small eigenvalue on the subspace $S\defeq \{\delta\in\R^n:\langle \delta,x\rangle=0\text{ and }\text{supp}(\delta)\subseteq F\}$. For this, we need to understand better the regularization construction -- as we will see in \pref{lem:taylor}, it follows from standard convex analysis arguments that
\[
    \nabla^2\omega^*(Ax)\preccurlyeq \text{diag}(\nabla)^{\frac{3}{2}}\quad \text{ for some vector $\nabla\in \Delta_n$}\mper
\]
By further use of the orthogonality trick, we can select a slightly smaller subspace than $S$ whose elements ``don't see'' the rows $i$ for which $\nabla_i\gtrsim 1/|F|$. A random element $\delta$ in this subspace achieves quadratic form at most $\|\delta\|_2^2/\sqrt{|F|}$ in expectation. The details can be found in \pref{lem:spencerlem}. This ultimately implies
\begin{align}
    \diff\omega^*(Ax)\lesssim \frac{\diff \|x\|_2^2}{\sqrt{|F|}}\mper\label{eq:analysis1}
\end{align}

\paragraph{Analysis of the whole dynamic.} For $k\in [n]$, denote by $t_k\defeq \min\{t\ge 0 : |F(t)|\le k\}$ the first time for which the number of active coordinates reaches $k$. From the constraint that the update direction is always orthogonal to the current partial coloring, we get after integrating \pref{eq:analysis1} over $t\in[0,T]$ that
\[
    \omega^*(Ax(T))-\omega^*(0)\lesssim \sum_{k=1}^n \frac{\|x(t_{k-1})\|_2^2-\|x(t_k)\|_2^2}{\sqrt{k}}.
\]
Finally, we apply summation by parts and use the fact that $\sum_{i\le k} \|x(t_{i-1})\|_2^2-\|x(t_i)\|_2^2\le k$:
\[
    \omega^*(Ax(T))-\omega^*(0)\lesssim\sqrt{n}+ \sum_{k=1}^{n-1} \frac{1}{k^{\frac{3}{2}}}\sum_{i\le k}  \|x(t_{i-1})\|_2^2-\|x(t_i)\|_2^2\lesssim \sqrt{n}\mper
\]
Combining this with our previous observation that $\omega^*(0)\lesssim \sqrt{n}$ concludes our proof sketch of \pref{thm:spencer}.

\subsection{Comparison with existing approaches}
\label{sec:comparison}

The general framework of tracking the discrepancy of a continuously evolving partial coloring through the means of a smooth approximation appears in the literature in various similar forms. Several other works on discrepancy use techniques that are similar to ours, the most important being those due to Lovett-Meka~\cite{LM15}, Levy-Ramadas-Rothvoss~\cite{LRR17}, and Bansal-Laddha-Vempala~\cite{BLV22}.
In what follows we provide a brief overview of these.

\paragraph{Comparison with~\cite{LM15}.}
Here, the authors give an algorithmic proof of Spencer's theorem, where they evolve a partial coloring  using a \textit{sticky} Brownian motion on the hypercube, and similarly to our case, they freeze coordinates once they reach $\pm 1$. Their algorithm, however, operates inside a convex set corresponding to low discrepancy colorings, and additionally requires maintaining several technical conditions. At the end of a phase, they only obtain a partial coloring, and need to repeat this procedure several times. 

\paragraph{Comparison with \cite{LRR17}.}
In this subsequent work, the authors  propose an elegant deterministic algorithm, inspired by Lovett-Meka as well as by the multiplicative weights update method. Their algorithm is very close in spirit to ours, as it iteratively updates a fractional coloring while controlling the exponential weights that are assigned to the set constraints. The exponential weights can be seen as a proxy for a smooth regularization of the maximum function. In fact, this approach is related to Spencer's \textit{hyperbolic cosine} algorithm~\cite{S85}, that essentially consists in tracking the evolution of $\sum_i \cosh(Ax)_i$, which is up to a reparametrization equivalent to our setting when using entropic regularization (see \pref{def:negentreg}). 

While the hyperbolic cosine algorithm is unable to obtain discrepancy below $O\left(\sqrt{n \log n}\right)$, the authors of \cite{LRR17} do so by combining it with the continuous approach from~\cite{LM15}. 
This approach alone does not directly manage to recover Spencer's bound. The reason can be  easily understood when interpreting their algorithm through the regularization perspective: the \textit{error} introduced by entropic regularization contains a logarithmic term which carries over to the final discrepancy bound. They observe, however, that they can force the approximation provided by their regularized maximum to be far from tight, in the sense that the true discrepancy is not as large as what the potential function ``sees''.
This property is enforced by taking sufficiently small steps, to ensure that at all times, all the rows attaining the largest discrepancy contribute equally. This forces the regularizer to spread a large part of its mass uniformly over these, and thus maximize the error it pays for beyond the true value of the maximum discrepancy.
Just like in~\cite{LM15}, this approach only obtains a partial coloring, and needs to be repeated with a different setting of parameters.

\paragraph{Comparison with \cite{BLV22}.} 
In this parallel work, the authors 
propose a unified approach for constructive discrepancy minimization using a barrier-based potential function. Their algorithm can be viewed as almost equivalent to ours, although the reason why the potential function works may appear quite magical.
Compared to~\cite{LRR17}, they replace exponential weights with a sum of inverse $p$ powers of slacks, where the slacks measure for each row the distance between a desired discrepancy upper bound and the current discrepancy.
In our regularization framework, a similar barrier emerges directly from choosing an appropriate regularizer. As we will soon see, we can derive it from first principles, and rather than having to guess a potential function and do tedious calculations to understand its evolution, we simply need to focus our attention on the trade-off between the error it introduces and the eigenvalues of its Hessian (see, e.g., our analysis in \pref{sec:spencer}).

\paragraph{Relation to regret minimization frameworks~\cite{BSS14, ALO15}.}
While both \cite{LRR17} and \cite{BLV22} rely on tracking a potential function, we attempt to make this approach more principled. The barrier potential present in~\cite{BLV22} appears to be related to the one employed by Batson-Spielman-Srivastava~\cite{BSS14} in the context of spectral sparsification. Interestingly, the reason a barrier was used in the  case of sparsification was exactly to remove an extra logarithmic factor that would have otherwise occurred when using standard entropic regularization/multiplicative weights. Allen-Zhu, Liao and Orecchia~\cite{ALO15}  made the connection between the barrier potential and the multiplicative weights method explicit by noticing that both follow from using different regularizers on top of the maximum function (although in their case they more generally regularize matrix norms). Note that in~\cite{ALO15}, the authors provide bounds on the second-order term of their regularized spectral norms in the form of multiplicative error on the gradient term. Here, we directly relate the second-order terms to the gradient, which allows us to obtain tighter bounds on the change in  our potential functions. We believe this to be of independent interest.

%% file: approach.tex
\section{The regularization framework}
\label{sec:approach}

\subsection{An iterative meta-algorithm}
\label{sec:generic}

We first describe a generic iterative algorithm for discrepancy minimization that will serve as a basis for incorporating the potential functions based on regularization. Similarly to \pref{sec:overview}, we will construct a sequence of partial colorings $x(t)\in[-1,1]^n$ for integer times $t=0,1, \ldots$. Each step consists in picking an update vector $\delta$ and adding it to $x(t)$. Whenever some coordinate of $x(t)$ becomes $\pm 1$, we say that the coordinate is \textit{frozen}. We will also say of an unfrozen coordinate that it is \textit{active}.

\begin{algorithm}
	\caption{Generic iterative algorithm for discrepancy minimization}
	\label{alg:iterative}
	\begin{algorithmic}[1]
		\STATE \textbf{Input:} $A\in \R^{m\times n}$, $L\in (0,1)$
		\STATE \textbf{Output:}  $x\in\{\pm 1\}^n$ (a low-discrepancy coloring of $A$)
		\STATE Let $x(0) \defeq (0, \ldots, 0)$ and $t\defeq 0$.
		\WHILE{$\text{oracle}(A,x(t))$ is not $\mathsf{undefined}$}
			\STATE Choose any unit vector $\delta$ in $\oracle(A,x(t))\,\cap\,\{\delta\in\R^n:\langle \delta,x(t)\rangle=0\}$.

			\STATE Let $\varepsilon(t)\defeq \min\{\varepsilon>0 : \exists j\in [n],\, x_j(t)\notin \{-1,1\}\text{ and } x_j(t)+\varepsilon\delta_j\in\{-1,1\}\}$.
			\STATE Set $x(t+1)\defeq x(t) + \min(L,\varepsilon(t))\delta$.
			\STATE Update $t\defeq t+1$.
		\ENDWHILE
		\STATE Let $T\defeq t$ and $x_j^*\defeq \text{sign}(x_j(T))$ for all $j\in [n]$.
		\RETURN $x^*$.
	\end{algorithmic}
\end{algorithm}

\paragraph{The oracle.} Suppose that we are given some blackbox algorithm $\mathsf{oracle}$ that encapsulates all the possible choices of directions of the update vector. In the sequel, $\oracle(A,x)$ will correspond to a subset of vectors that do not increase too much the value of the regularized potential function when $x$ is the current partial coloring. 

\begin{assumption}
    \label{ass:oracle}
    Let $C>0$ be some universal constant. Given a matrix $A\in\R^{m\times n}$ and a partial coloring $x\in [-1,1]^n$, $\oracle(A,x)$ satisfies (with $F\defeq \{j\in [n] : x_j\notin \{-1, 1\}\}$):
    \begin{itemize}
        \item If $|F|\ge C$, $\oracle(A,x)$ is a subset of $\R^F$ such that the intersection of $\oracle(A,x)$ with any halfspace of $\R^F$ contains a half-line.
        \item If $|F|<C$, it returns the value $\mathsf{undefined}$.
    \end{itemize}
\end{assumption}

\paragraph{} With $\oracle$ being given, the meta-algorithm for discrepancy minimization is described as \pref{alg:iterative}. The following three immediate observations on \pref{alg:iterative} will be central to our framework.

\begin{observation}
    For any $t=0,\ldots,T-1$, $\|x(t+1)-x(t)\|_\infty\le L$.
\end{observation}

\begin{observation}
    \label{obs:disc}
    The final step on line 9 adds at most $C\max_{i\in [m],j\in[n]} |A_{ij}|$ to the discrepancy of the coloring, where $C$ is the constant from \pref{ass:oracle}.
\end{observation}

\begin{observation}
    \label{obs:comp}
	There can be at most $n/L^2$ iterations of the main loop of \pref{alg:iterative}. Therefore, \pref{alg:iterative} runs in polynomial time as long as $\oracle$ runs in polynomial time and $L\ge n^{-O(1)}$.
\end{observation}

\begin{proof}
    When $\varepsilon(t)\le L$, at least one additional coordinate will reach $\pm 1$ and will be frozen at the end of the iteration. This can happen at most $n$ times. When $\varepsilon(t)>L$, since we pick our update vector orthogonal to $x$, we have $\|x(t+1)\|_2^2=\|x(t)\|_2^2+L^2$. This can happen at most $n/L^2$ times.
\end{proof}

For our purposes, we will always set $L=n^{-O(1)}$ and computing $\oracle$ will only require elementary linear algebraic operations in $\R^n$ (intersection, orthogonal complements, direct sums, computation of eigenspaces, etc.).

\subsection{Regularized maximum}
\label{sec:potential}

Our main tool for building proxies for discrepancy is the following regularized version of the maximal entry of a vector.

\begin{definition}
    For any convex function $\phi:\Delta_m\to\R$, we define $\phi^*:\R^m\to\R$ by
    \[
        \phi^*(y)\defeq \max_{r\in\Delta_m} \langle r,y\rangle-\phi(r)\mper
    \]
\end{definition}

We will call $\phi$ the \textit{regularizer} -- it maps elements of the simplex to some penalty in a convex way. By symmetry, it makes sense to focus on regularizers of the form $\phi(r)=\sum_{i\in [m]} \varphi(r_i)$ for some convex $\varphi:\R\to \R$. The following two special cases will play an important role in our theory.

\begin{definition}
    For any $0<q<1$, the \textit{$\ell_q$-regularization} of the maximum, parametrized by $\eta>0$, is the function $\omega^*_{q,\eta}:\R^m\to\R$ such that
	\[
		\omega^*_{q,\eta}(y)\defeq \max_{r\in\Delta_m}\, \langle r,y\rangle+\frac{1}{\eta q}\sum_{i=1}^m r_i^{q},\quad \text{ for any $y\in\R^m$}\mper
	\]
\end{definition}

\begin{definition}\label{def:negentreg}
	The \textit{(negative) entropy regularization} of the maximum, parametrized by $\eta>0$, is the function $\smax:\R^m\to\R$ such that
	\[
		\smax_\eta(y)\defeq \max_{r\in\Delta_m}\, \langle r,y\rangle-\frac{1}{\eta}\sum_{i=1}^m r_i\log r_i,\quad \text{ for any $y\in\R^m$}\mper
	\]
	It is not hard to see that in this case, the solution of the maximization problem can be written in closed form: $\smax_\eta(y)=\frac{1}{\eta}\log\left(\sum_{i=1}^m \exp(\eta y_i)\right)$, thereby recovering the usual formulation of the softmax function.
\end{definition}

\paragraph{Regret minimization interpretation.} These regularization ideas have originated in the online learning community. To see how this is related to discrepancy, let us make the following thought experiment. We play an online game against an adversary, where we select at each step some $r_t\in\Delta_n$, and after that some $\delta_t$ is revealed. Our goal is to minimize the regret, which is the difference between the best static cost in hindsight, $\| A(\delta_1+\ldots+\delta_T)\|_\infty$ and our cost, $\sum_{t\le T} \langle r_t, A\delta_t\rangle$. We can play different strategies that are robust to the future choices of $\delta_t$. For each of these, the optimal strategy for the adversary is to follow a particular potential function that is obtained by adding a regularizer to the optimization problem.

We note that this is analogous to what Allen-Zhu, Liao, and Orecchia have proposed for graph sparsification \cite{ALO15}. They identified a similar online game on density matrices and used it to interpret the construction of \cite{BSS14} as a follow-the-regularized-leader strategy. While they use the same $\ell_{1/2}$-regularizer as we did for Spencer's theorem in \pref{sec:overview}, 
the connection is more subtle, as their setting is crafted specifically for matrices with positive updates, which involves deriving a set of different bounds that charge the entire change in potential function to the first-order term.

\subsection{Regularization bounds}
\label{sec:bounds}

We now present our two main technical lemmas that give an analytic justification for the $\ell_q$ and negative entropy regularization. The first one (\pref{lem:explain}) estimates the additive error incurred when tracking the regularized version of the maximum instead of the true maximum. For constant $\eta$ and $q$, the approximation is worse for $\ell_q$-regularization than for negative entropy regularization (polynomial vs logarithmic in the size of the vector). 

\begin{lemma}
	\label{lem:explain}
	Let $y\in\mathbb{R}^m$ and $q\in (0,1)$. If $M(y) \defeq \max_{1\le i\le m} y_i$,
	\[
		M(y)\le \omega^*_{q,\eta}(y)\le M(y)+\frac{m^{1-q}}{\eta q} \text{ and } M(y)\le \smax_\eta(y)\le M(y)+\frac{\log m}{\eta}\mper
	\]
\end{lemma}

\begin{proof}
	The lower bounds follow from picking $r$ to be the Dirac mass function centered on the maximum coordinate. For the upper bounds, note that on the one hand, for all $r\in\Delta_m$, $\langle r,y\rangle \le M(y)$, and on the other hand, $\sum_i r_i^q\le m^{1-q}$ (resp. $-\sum_i r_i\log r_i\le \log m$) by Jensen's inequality.
\end{proof}

The second one (\pref{lem:taylor}) bounds the first two terms in the Taylor expansion of the potential function. In the sequel, this will allow us to control the increase in $\ell_\infty$-norm when making a small update in our iterative algorithm. As we  demonstrated in \pref{sec:overview}, what matters in this expansion is the second-order term. Indeed, in applications to discrepancy, we will always trivially upper bound the first-order term by simply picking an update that is positively correlated with the gradient (which will be an easy additional condition to impose).

\begin{lemma}
\label{lem:taylor}
    Fix $y\in\R^m$ and $q\in (0,1)$. Let $\nabla\defeq \nabla \omega^*_{q,\eta}(y)$. Then $\nabla\in \Delta_m$ and for all $\delta\in\R^m$ with $\|\delta\|_\infty\le \frac{1-q}{8\eta}$,
	\[
		\omega_{q,\eta}^*(y+\delta)\le \omega_{q,\eta}^*(y)+\langle \nabla, \delta\rangle+\frac{\eta}{1-q} \sum_{i=1}^m \nabla_i^{2-q} \delta_i^2\mper
	\]
	Similarly, if $\nabla\defeq \nabla \smax_\eta(y)$, then $\nabla\in\Delta_m$ and for all $\delta\in\mathbb{R}^m$ with $\lVert \delta\rVert_\infty\le \frac{1}{3\eta}$,
	\[
		\smax_\eta(y+\delta)\le \smax_\eta(y)+\langle \nabla, \delta\rangle+\eta \sum_{i=1}^m \nabla_i \delta_i^2\mper
	\]
\end{lemma}

\begin{proof}
    Consider first the $\ell_q$-regularizer with $\eta=1$. To lighten notations we write $\omega^*$ for $\omega^*_{q,\eta}$. Recall that
    \begin{align}
        \omega^*(y)=\max_{r\in \Delta_m} \langle r,y\rangle+\frac{1}{q}\sum_{i=1}^m r_i^q\mper\label{eq:recall}
    \end{align}
    By Danskin's theorem (see e.g. \cite[Proposition B.25]{B99}), we have $\nabla \omega^*(y)=r^*\in \Delta_m$, where $r^*$ is the optimum in \pref{eq:recall}. For the KKT conditions to hold, we must have for some $\lambda:\R^m\to \R$ (the Lagrange multiplier associated to the equality constraint of the simplex): $y_i+(r_i^*)^{q-1}=\lambda(y)$ for all $i\in [m]$. Note that the Lagrange multipliers associated to the inequality constraints disappear by complementary slackness since necessarily $r_i^*\neq 0$. Also we must have $\lambda(y)>\max_{i\in [m]} y_i$ by the previous equality. In fact, $\lambda(y)$ is the unique solution to $\sum_{i\in [m]} (\lambda(y)-y_i)^{1/(q-1)}=1$.
    
    In summary, $\nabla\omega^*(y)=(\lambda(y) \mathbf{1} - y)^{\odot \frac{1}{q-1}}\in\Delta_m$. Differentiating once more, we see that
	\[
		\nabla^2 \omega^*(y)=\frac{1}{1-q}\left(\diag(\nabla\omega^*(y)^{\odot 2-q})-(\nabla \lambda(y))(\nabla\omega^*(y)^{\odot 2-q})^\top\right)\mper
	\]
    Let $M\defeq(\nabla \lambda(y))(\nabla\omega^*(y)^{\odot 2-q})^\top$. Observe that M has rank 1 and must be symmetric as the Hessian itself is symmetric. Further, $\lambda(y)$ is a nondecreasing function of $y_i$ for all $i\in[m]$, so that every entry of $M$ is nonnegative. It follows that $M$ is positive semidefinite, and thus
    \begin{align}
        \label{eq:gradhes}
        \nabla^2 \omega^*(y)\preccurlyeq \frac{1}{1-q}\diag(\nabla \omega^*(y)^{\odot 2-q})\mper
    \end{align}

    Now fix $\delta\in\mathbb{R}^m$. The function $s\mapsto \sum_i (s-y_i)^{\frac{1}{q-1}}$ defined for $s>\max_i y_i$ is nonincreasing, so for all $i$,
    \begin{equation}
        \label{eq:lambdainf}
        |\lambda(y+\delta)-\lambda(y)|\le \lVert \delta\rVert_\infty \text{ and } \lambda(y)\ge 1+y_i\mper
    \end{equation}
    Now fix $i\in[m]$ and suppose that $\delta$ satisfies $\lVert \delta\rVert_\infty\le \frac{1-q}{8}$. We write
	\begin{align*}
		(\nabla\omega^*(y+\delta))_i^{2-q}
		&=(\nabla \omega^*(y))_i^{2-q}\left(1+\frac{\lambda(y+\delta)-\lambda(y)-\delta_i}{\lambda(y)-y_i}\right)^{\frac{2-q}{q-1}}\\
		&\le (\nabla \omega^*(y))_i^{2-q}\exp\left(\frac{2-q}{1-q} \frac{\lambda(y)+\delta_i-\lambda(y+\delta)}{\lambda(y+\delta)-y_i-\delta_i} \right)\mcom
	\end{align*}
    where we used the inequality $\log(1+y)\ge \frac{y}{1+y}$. Now we plug in the inequalities \pref{eq:lambdainf}:
	\begin{align}
		(\nabla\omega^*(y+\delta))_i^{2-q}
		&\le (\nabla \omega^*(y))_i^{2-q}\exp\left(\frac{2-q}{1-q}\frac{\lambda(y)+\delta_i-\lambda(y+\delta)}{1+\lambda(y+\delta)-\lambda(y)-\delta_i}\right)\nonumber\\
		&\le (\nabla \omega^*(y))_i^{2-q}\exp\left(\frac{2}{1-q}\frac{2\|\delta\|_\infty}{1-2\|\delta\|_\infty}\right)\nonumber\\
		&\le 2(\nabla \omega^*(y))_i^{2-q}\mper\label{eq:ineqgrad}
	\end{align}
    Finally, from Taylor's inequality, under the same assumption $\lVert \delta\rVert_\infty\le \frac{1-q}{8}$,
	\[
		\left|\omega^*(y+\delta)- \omega^*(y)-\langle \nabla\omega^*(y),\delta\rangle\right|\le \frac{1}{2}\sup_{u\in[0,1]} \left|\delta^\top \nabla^2 \omega^*(y+u\delta)\delta\right|\le \frac{1}{1-q}\sum_{i=1}^m (\nabla\omega^*(y))_i^{2-q} \delta_i^2\mcom
	\]
    where the last inequality follows from \pref{eq:gradhes} and \pref{eq:ineqgrad}.
    
    For the entropy regularizer and $\eta=1$, it holds that
    \[
        \nabla \smax(y)=\frac{\exp(y)}{\sum_{i=1}^m \exp(y_i)}\text{ and } \nabla^2\smax(y)\preccurlyeq \diag(\nabla\smax(y))\mper
    \]
    Therefore for all $i\in [m]$, $(\nabla\smax(y+\delta))_i\le (\nabla\smax(y))_i \exp(2\lVert\delta\rVert_\infty)$, and we conclude in the same way as for the $\ell_q$-regularizers.
    
	For general $\eta$, observe that $\nabla \omega_{q,\eta}^*(y)=\nabla\omega_{q,1}^*(\eta y)$ and $\nabla^2\omega_{q,\eta}^*(y)=\eta\nabla^2\omega_{q,1}^*(\eta y)$ (and similarly for $\smax_\eta$). Therefore, the same argument based on Taylor's inequality gives the desired result as long as $\lVert\delta\rVert_\infty\le \frac{1-q}{8\eta}$ for $\omega^*_{q,\eta}$ and $\lVert\delta\rVert_\infty\le \frac{1}{3\eta}$ for $\smax_\eta$. 
\end{proof}

\begin{remark}
    This gives an analytic explanation for why we might prefer $\ell_q$-regularization to negative entropy regularization in certain situations, although the approximation error from \pref{lem:explain} is worse (for the same value of $\eta$). Observe that a typical entry $\nabla_i$ of the gradient is much smaller that $1$. Hence $\ell_q$-regularization can be advantageous whenever we can leverage the fact that $\nabla_i^{2-q}$ is typically much smaller than $\nabla_i$. As we will now see, this is the case in Spencer's setting.
\end{remark}

%% file: spencer.tex
\section{Spencer's setting}
\label{sec:spencer}

We now focus on the setting of Spencer's theorem (\pref{thm:spencer}), namely the discrepancy of matrices with bounded entries. Our goal in this section is twofold: first, we give a rigorous version of the proof of Spencer's theorem that we sketched in \pref{sec:overview}. Then, we show how to improve the constant with a slightly more careful analysis.

\subsection{Full proof of Spencer's theorem}
\label{sec:spencereasy}

We give a complete proof of Spencer's theorem in the general case where the matrix has $m$ rows and $n$ columns. Our choice of $q\in(0,1)$ in the $\ell_q$-regularization is going to depend on the ratio $m/n$.

\begin{theorem}
	\label{thm:spencergeneral}
	Let $n\le m$. There is a deterministic algorithm running in polynomial time that for each $A\in[-1,1]^{m\times n}$, finds $x\in\{\pm 1\}^n$ such that
	\[
	    \lVert Ax\rVert_\infty = O\left(\sqrt{n\log(\frac{2m}{n})}\right).
	\]
\end{theorem}

We start by proving the following lemma, which will allow us to find an update vector that does not increase too much the $\ell_q$-regularization of the maximal coordinate when there are $k$ active coordinates remaining.

\begin{lemma}
    \label{lem:spencerlem}
    Let $k, m$ be such that $4\le k\le 2m-2$. Let $u_1, \ldots, u_m$ be unit vectors in $\R^k$ and $\nabla\in\Delta_m$. Consider
    \[
        M\defeq \sum_{i=1}^m \nabla_i^{2-q} u_i u_i^\top.
    \]
    There is a subspace $S$ of dimension at least 2 such that for all $v\in S$,
    \[
        v^\top Mv\le 8k^{q-2}\|v\|_2^2.
    \]
    Moreover, this subspace can be found efficiently.
\end{lemma}

\begin{proof}
    Without loss of generality, suppose that $\nabla_1\ge \ldots \ge\nabla_m$.
    Let
    \[
        S_1\defeq \left\{v\in\R^k : \langle v,u_i\rangle=0 \text{ for all } i=1,\ldots,\left\lceil\frac{k}{2}\right\rceil-1\right\}.
    \]
    Observe that $\nabla_{\left\lceil\frac{k}{2}\right\rceil}\le \frac{2}{k}$, so for all $v\in S_1$,
    \begin{align}
        \label{eq:spencerfirst}
        v^\top M v\le \left(\frac{2}{k}\right)^{1-q} \sum_{i=1}^m \nabla_i \langle u_i, v\rangle^2.
    \end{align}
    Let $R\defeq \sum_{i\le m} \nabla_i u_i u_i^\top$ and consider an orthonormal basis $w_1, \ldots, w_\ell$ of $S_1$ such that $w_1^\top Rw_1\le \ldots \le w_\ell^\top R w_\ell$. Select $S$ to be the span of $\{w_1,w_2\}$. Observe that
    \[
        \sum_{j=1}^\ell w_j^\top R w_j=\sum_{i=1}^m \nabla_i \sum_{j=1}^\ell \langle w_j,u_i\rangle^2\le 1.
    \]
    Thus, by an averaging argument, it holds that
    \[
        v^\top R v\le \frac{1}{\ell-1}\|v\|_2^2\le \frac{2}{k-2}\|v\|_2^2
    \]
    for all $v\in S$. We conclude by combining this with \pref{eq:spencerfirst} and using the assumption on $k$.
\end{proof}

\begin{proof}[Proof of {\pref{thm:spencergeneral}}]
    We first double all the rows of $A$ and consider the matrix $\begin{bmatrix} A \\-A  \end{bmatrix}$. Thus, we assume without loss of generality that we are given a $2m\times n$ matrix $A$ such that for all $x\in\hc$, $\|Ax\|_\infty=\max_i (Ax)_i$.

    We set the parameter $L$ of \pref{alg:iterative} to be $L\defeq \frac{1-q}{8\eta n}$, where $q$ and $\eta$ are the parameters of the $\ell_q$-regularizer to be fixed later.
    
    We now describe our construction of $\oracle(A,x(t))$ with $F(t)$ being the set of active coordinates of $x(t)$ and $k=k(t) \defeq |F(t)|$. To simplify notations, we write $x=x(t)$ and $F=F(t)$. Observe that $\|A\delta\|_\infty\le n\|\delta\|_\infty$. Hence, by \pref{lem:taylor}, there exists $\nabla\in\Delta_n$ such that for all update $\delta\in\mathbb{R}^n$ with $\lVert \delta\rVert_\infty\le L$,
	\[
		\omega^*_{q,\eta}(A(x+\delta))-\omega^*_{q,\eta}(Ax)\le \langle \nabla, A\delta\rangle+\frac{\eta}{1-q}\sum_{i=1}^n \nabla_i^{2-q} \langle A_i,\delta\rangle^2.
	\]
	By assumption, $\sum_{j\in F} A_{i,j}^2\le k$, so we can apply \pref{lem:spencerlem} to get a 2-dimensional subspace $S$ such that for all $\delta\in S$,
	\begin{align}
	    \label{eq:incr}
	    \sum_{i=1}^n \nabla_i^{2-q} \langle A_i,\delta\rangle^2\le
	    \frac{4\|\delta\|_2^2}{k^{1-q}}.
	\end{align}
	The second-order term is invariant if we change $\delta$ to $-\delta$, but the first-order term changes sign. We return from $\oracle(A,x(t))$ the subspace $S$ intersected with the halfspace $\{\delta\in\R^n : \langle A^\top \nabla,\delta\rangle\le 0\}$.

    Now we switch to the global analysis of \pref{alg:iterative} and estimate what is the total discrepancy incurred over the whole walk. \pref{ass:oracle} is here satisfied for $C=3$, so since the entries of $A$ are bounded, the last step of \pref{alg:iterative} only changes the discrepancy of the final coloring by an additive constant.
    
    Denote by $\beta_k$ the sum of the $\ell_2$-squared norm of the update vectors starting from the point where there are at most $k$ unfrozen coordinates remaining.  Recall that we always choose our update vector orthogonal to the current position, so that $\beta_{k}\le k$. We now sum by parts the main term of the increases \pref{eq:incr} over the execution of the algorithm,
	\begin{align}
		\sum_{k=4}^n \frac{\beta_k-\beta_{k-1}}{k^{1-q}}&=\frac{\beta_n}{n^{1-q}}+\sum_{k=4}^{n-1}
		\beta_k\left(\frac{1}{k^{1-q}}-\frac{1}{(k+1)^{1-q}}\right)
		\le n^{q}+\sum_{k=4}^{n-1} \frac{1}{k^{1-q}}
		\le \frac{2n^q}{q}.\label{eq:final}
	\end{align}
	Thus, by \pref{eq:final} and \pref{lem:explain}, the final coloring $x(T)$ satisfies
	\[
		\lVert Ax(T)\rVert_\infty\le \omega^*_{q,\eta}(Ax(T))\le \frac{(2m)^{1-q}}{\eta q}+\frac{8\eta}{q(1-q)}n^{q}.
	\]
	The result follows by setting
	\[
	    \eta=\sqrt{\frac{(1-q)m^{1-q}}{n^q}} \text{ and } q=1-\frac{1}{\log(\frac{2m}{n})}.\qedhere
	\]
\end{proof}

\begin{remark}
    At this point, it is worth looking at what happens in this proof if we replace the $\ell_q$-regularizer with the entropic regularizer. For simplicity, consider the case where $m=O(n)$. While the constant cost is only $\log n/\eta$, we are not able to win anything in the local update as in \pref{lem:spencerlem} and we would get an $\eta n$ loss in the potential during the walk. Optimizing over $\eta$ would give discrepancy $\sqrt{n\log n}$. Negative entropy regularization in this context corresponds merely to a derandomization of the Chernoff and union bound argument.
    
    In fact, one could repeat the same analysis by replacing the regularizer by a general function of the form $\phi(r)=\sum_{i\le n} \varphi(r_i)$ for some convex, non-positive function $\varphi:\R\to\R$. Under additional conditions on $\varphi$ (for example the fact that $x\mapsto x\varphi''(x)$ is non-increasing) one would obtain a  discrepancy of
    \begin{equation}
        O\left( \sqrt{-n \varphi\left(\frac{1}{n}\right) \sum_{k\le n} \frac{k}{\varphi''\left(\frac{1}{k}\right)}} \right).\label{eq:discdiff}
    \end{equation}
    With this bound established, we can quickly verify that setting $\varphi$ to be the negative entropy, we obtain $\varphi(1/n) = -\log n/n$ and $\varphi''(1/k) = k$, which immediately recovers a discrepancy of $O(\sqrt{n\log n})$.
    
    Given this expression in~\pref{eq:discdiff}, it appears that we can derive the best possible regularizer by solving a differential equation. Since there is no silver bullet for such problems, one can simply test various elementary functions. Setting $\varphi(x) = -x^q$ for $0<q<1$ we verify the required condition and obtain $\varphi(1/n) = -1/n^q$, and $\varphi''(1/k) = q(1-q)/k^{2-q}$, which removes the logarithmic factor for constant $q$.
\end{remark}

\begin{remark}[Spherical discrepancy] 
    A slight variation of the same algorithm, which does not freeze variables, automatically achieves optimal bounds for \textit{spherical discrepancy}. This setting is a relaxation of the \Komlos problem, where the columns of the input matrix are vectors with at most unit $\ell_2$-norm, but the sought coloring only has an $\ell_2$-norm constraint i.e. $\|x\|_2 = \sqrt{n}$, rather than $x \in \{\pm 1\}^n$~\cite{jones2020spherical}. The key difference between this setting and that of \Komlos is that we are not forced to lose degrees of freedom by freezing variables, so throughout the entire execution of the algorithm we have $\Theta(n)$ degrees of freedom to update the partial coloring.

    To show this, we simply observe that at all times there is an update that does not increase the discrepancy of rows with large \textit{global} $\ell_2$ norm, which represent only at most a constant fraction of the entire set of rows, by Markov's inequality. The rate of increase in discrepancy entirely depends on the $\ell_2$-norm of the rows of the underlying matrix (restricted to the unfrozen variables, which in this case are all the variables). Following through with the same argument we used for Spencer, we obtain discrepancy $O(1)$.
\end{remark}

\subsection{Towards a better constant for Spencer's theorem}
\label{sec:spencerconstant}

In his original paper, Spencer proves that any matrix in 
$[-1,1]^{n\times n}$ has discrepancy at most $5.4 \sqrt n$~\cite{S85}.
In this section, we improve Spencer's bound to $4.1\sqrt n$. 
Prior to this work,~\cite[\S 5]{B13} improves Spencer's
bound to $5.2\sqrt n$ and sketches how to obtain
$3.7\sqrt n$, but some
computations rely on personal communication.
Unlike all these previous results, our proof is algorithmic. 

\begin{theorem}
	\label{thm:spencerConstant}
	For every $A\in[-1,1]^{n\times n}$, there exists $x\in\{-1,1\}^n$ such that 
	\[\lVert A x\rVert_\infty \le 4.1 \sqrt n + O(1)\,.\]
	Moreover, $x$ can be found by a randomized algorithm 
	running in polynomial time.
\end{theorem}

To prove~\pref{thm:spencerConstant}, we revisit the argument 
from~\pref{sec:spencer} by tracking constants more carefully.
We start by giving an analog of~\pref{lem:spencerlem} with 
a tighter leading constant.

\begin{lemma}
    \label{lem:spencerlem2}
	There exists $C>0$ such that the following holds
	for any $q\in (0,1)$ and $k,n\ge 1$.
	Let $M = \sum_{i=1}^n \nabla_i^{2-q}  u_i u_i^\top$
	for some vectors $u_1, \ldots, u_n$ in the unit
	ball of $\R^k$,
	and $\nabla\in\Delta_n$.
	
	Then there exists a 2-dimensional subspace $S$ such that the
	projection
	$\Pi_S$ onto $S$ satisfies
	\[
		\| \Pi_S M \Pi_S\|_{\textnormal{op}}\le \left(1 + \frac C k\right) \frac 1 {k^{2-q}}\,.
	\]
	Moreover, for any constant $\eps>0$, there is a randomized 
	polynomial-time algorithm outputting a 2-dimensional subspace
	$S$ satisfying with high probability
	\[
		\|\Pi_S M \Pi_S\|_{\textnormal{op}}\le (1+\eps)\left(1 + \frac C k\right) \frac 1 {k^{2-q}}\,.
	\]
\end{lemma}

\begin{proof}
	Assume that $\nabla_1\ge \ldots \ge \nabla_m$ without loss 
	of generality.
	We will prove that if $\alpha$ is sampled 
	uniformly in the interval $(\frac{1}{2},1)$, then
    \begin{align}
        \label{eq:randalpha}
	f(\nabla)\defeq \E_{\alpha\sim (\frac12, 1)}\left[ \sum_{i\ge \lfloor \alpha k\rfloor} \nabla_i^{2-q}\right]\le \frac{k^{q-1}}{4}+O(k^{q-2})=\E_{\alpha\sim (\frac12, 1)}[(1-\alpha)k^{q-1}]+O(k^{q-2})\,.
    \end{align}
	We first explain why~\pref{eq:randalpha} implies the desired bound.
	Let $\alpha\in(\frac{1}{2},1)$ be such that \[\frac{1}{(1-\alpha)k}\sum_{i\ge \lfloor \alpha k\rfloor} \nabla_i^{2-q}\le k^{q-2}+O(k^{q-3})\,.\] Then we can repeat the proof of~\pref{lem:spencerlem} to get a 2-dimensional subspace $S$ such that for all unit $v\in S$,
    \[
		\langle v, M v\rangle \le k^{q-2}+O(k^{q-3})\,,
    \]
	which is equivalent to the desired statement since $M\succeq 0$.
    Furthermore, if $\varepsilon>0$ is constant, the corresponding 
	$\alpha$ can be found with high probability by repeating 
	the experiment and using Markov's inequality.

    It remains to prove~\pref{eq:randalpha}. We compute explicitly
    \[
        f(\nabla)=2\int_{\frac{1}{2}}^1 \sum_{i\ge \lfloor\alpha k\rfloor} \nabla_i^{2-q}\diff{\alpha}= \sum_{i\ge \left\lfloor\frac{k}{2}\right\rfloor} \left(\frac{2(i+1)}{k}-1\right)\nabla_i^{2-q}\,.
    \]
    Let $\mathcal P=\{\nabla\in\Delta_m:\nabla_1\ge \ldots \ge \nabla_m\}$. Since $f\colon \mathcal P\to \R$ is a convex function, it attains its maximum at an extreme point of $\mathcal P$. Those are of the form $z_\ell\defeq \left(\frac{1}{\ell}\ldots \frac{1}{\ell}\,0\ldots 0\right)$ (with $\ell$ nonzero coordinates) for some $\ell\in [m]$. Moreover, the maximum of $f$ has to be attained when $\ell=\gamma k$, with $\gamma\in[\frac{1}{2},1]$. However, in that case,
	\begin{align*}
		f(z_\ell)&=\ell^{q-2}\left(\frac{\ell(\ell+1)-\frac{k}{2}(\frac{k}{2}+1)}{k}-\left(\ell-\frac{k}{2}+1\right)+O(1)\right)\\
			  &= k^{q-1}\left(\sqrt{\gamma}-\frac{1}{\sqrt{\gamma}}+\frac{1}{4\gamma^{3/2}}\right)+O(k^{q-2})\,.
	\end{align*}
	Finally, $\gamma\mapsto \sqrt{\gamma}-\frac{1}{\sqrt{\gamma}}+\frac{1}{4\gamma^{3/2}}$ is increasing on $[\frac{1}{2},1]$, with maximum equal to $\frac{1}{4}$ for $\gamma = 1$. This concludes the proof of~\pref{eq:randalpha}.
\end{proof}

We also sharpen the constant in front of the second-order term in~\pref{lem:taylor}.

\begin{lemma}
    \label{lem:taylorplus}
    There exist universal constants $C_1,C_2\in (0,1)$ such that 
	if $\nabla\defeq \nabla \omega^*_{q,\eta}(y)$, then for all $\delta\in\R^n$ with $\|\delta\|_\infty\le C_1\frac{1-q}{n\eta}$,
	\[
		\omega_{q,\eta}^*(y+ \delta)\le \omega_{q,\eta}^*( y)+\langle  \nabla,  \delta\rangle+\frac{\eta}{2(1-q)}\left(1+\frac{C_2}{n}\right) \sum_{i=1}^n \nabla_i^{2-q} \delta_i^2.
	\]
\end{lemma}

\begin{proof}
    The proof is identical to the proof of \pref{lem:taylor}. We simply replace \pref{eq:ineqgrad} by the stronger inequality following from the stronger assumption on $\|\delta\|_\infty$.
\end{proof}

\begin{proof}[Proof of~{\pref{thm:spencerConstant}}]
   	We follow the proof of~\pref{thm:spencergeneral}. We set the 
	update size in \pref{alg:iterative} to be $L=\frac{C_1}{4n^2}$, 
	where $C_1$ is the constant from \pref{lem:taylorplus}. 
	We use the doubling trick to replace $ A$
	by a $2n\times n$ matrix such that $\| A  x\|_\infty = \max_{i\in [n]} \langle  A_i,  x\rangle$ for any vector $ x$.
	We use the potential function $\omega^*_{q,\eta}$ for some 
	parameters
	$q\in (0,1)$ and $\eta>0$ to be optimized at the end.

	First, by~\pref{lem:explain}, the initial loss is
	\[
		\omega^*_{q,\eta}( 0)\le \frac {(2n)^{1-q}} {\eta q}\,.
	\]
    
	Then, at any time $t$, we apply~\pref{lem:taylorplus} to 
	get that for any $\| \delta\|_\infty\le L$,
    \[
		\omega^*_{q,\eta}( A x(t)+ A \delta)-\omega^*_{q,\eta}( A x(t))\le \langle  A  \delta,  \nabla\rangle + \frac{\eta}{2(1-q)}\left(1+\frac{C_2}{n}\right) \sum_{i=1}^{2n} \nabla_i^{2-q} \langle  A_i, \delta\rangle^2\,,
    \]
	where $ \nabla = \nabla \omega^*_{q,\eta}( A  x(t))\in \Delta_{2n}$.
    Next, we apply~\pref{lem:spencerlem2}, where the $ u_i$ are the
	normalized rows of $ A$ restricted to the active coordinates. 
	Note that these rows have $\ell_2$-norm at most $\sqrt k$, where
	$k$ is the number of active coordinates.
	In this way, we find a $2$-dimensional 
	subspace $S$ such that if $ \delta\in S$ has small enough
	$\ell_\infty$-norm,
    \[
		\frac 1 {\| \delta\|_2^2}\sum_{i=1}^{2n} \nabla_i^{2-q} \left\langle \frac 1 {\sqrt k}  A_i, \delta\right\rangle^2\le \left(1 + \frac C k\right) \frac 1 {k^{2-q}}
	\]
	We then choose a direction for $ \delta\in S$ so that
	$\langle  \delta,  x(t)\rangle = 0$, and a
	signing  $\pm  \delta$ that makes the first-order term
	$\langle  A (\pm  \delta), \nabla\rangle < 0$. With this
	choice of $ \delta$, the increase in the potential at 
	time $t$ is at most
	\[
		\omega^*_{q,\eta}( A x(t)+ A \delta)-\omega^*_{q,\eta}( A x(t))\le \| \delta\|_2^2 \,\frac {\eta} {2(1-q)} \left(1 + \frac {C_2} n\right) \left(1 + \frac C k\right) \frac 1 {k^{1-q}}\,.
	\]
	Then, the summation by parts argument of~\pref{sec:spencer}
	yields
	\begin{align*}
		\| A  x(T)\|_\infty&\le \omega^*_{q,\eta}( A  x(T))\\
								 &\le \omega^*_{q,\eta}( 0) + \sum_{t=1}^{T} \omega^*_{q,\eta}( A  x(t)) - \omega^*_{q,\eta}( A  x(t-1))\\
										 &\le \frac {(2n)^{1-q}} {\eta q} + \frac {\eta} {2(1-q)} \cdot \frac {n^{q}} {q} + O(1)\,. 
	\end{align*}
	Optimizing over $\eta>0$ and $q\in (0,1)$, we obtain that up to
	a constant additive error, the
	discrepancy of $ A$ is at most
	\[
		2 \sqrt n \min_{q\in (0,1)} \sqrt{\frac 1 {2^{q} q^2 (1-q)}} \le 4.1 \sqrt n \,,
	\]
	achieved for $q \approx 0.71$.
	Finally, the algorithmic statement follows from using 
	the constructive version of \pref{lem:spencerlem2}.
\end{proof}

The best known asymptotic lower bound on the constant in Spencer's
problem, as $n\to\infty$, is $1$, achieved by Hadamard matrices (and by
random matrices).
In general, we do not expect this analysis
to be tight.
In particular, the use of $\ell_{0.71}$-regularization seems somewhat
ad hoc, and the analysis disregards information contained in the
first-order term.

%% file: pseudorandom.tex
\section{New discrepancy bounds for pseudorandom instances}
\label{sec:main}

An advantage of approaches for discrepancy minimization based on potential functions is that it is easy to track two (or a constant number of) potentials in parallel. Up to changing a few constants in the analysis, we can ensure that the best of both worlds happens. In this section, we illustrate this idea by proving a new discrepancy bound for instances of Beck-Fiala and \Komlos conjectures that satisfy a certain pseudorandomness condition.

Following \cite{P20}, we define the quantity $\lambda(A)$ associated to a matrix $A\in\R^{m\times n}$ as follows:
\[
    \lambda(A)\defeq \sup_{\|u\|_2=1,\langle u,\mathbf 1\rangle=0} \|B u\|_2,\text{ where $B_{ij}\defeq A_{ij}^2$ for all $i\in [m],j\in [n]$}\mper
\]
In the special case where $A$ is the adjacency matrix of a $d$-regular graph, $\lambda(A)$ is the second largest eigenvalue of $A$ and is bounded by $d$. As observed in \cite{P20}, $\lambda(A)$ is typically much smaller than this worst-case bound when $A$ is the incidence matrix of a \textit{random} regular set system. We will also check in \pref{sec:corollary} that this still holds for natural random instances of \Komlos conjecture.

We now recall Potukuchi's result for pseudorandom Beck-Fiala instances.

\begin{theorem}[Theorem 1.1 in \cite{P20}]
	\label{thm:potu}
	Let $A\in\{0,1\}^{m\times n}$ be such that each column has at  most $s$ nonzero entries. Then there is a randomized algorithm running in polynomial time to find $x\in\{\pm 1\}^n$ such that
	\[
	    \lVert Ax\rVert_\infty=O(\sqrt{s}+\lambda(A))\mper
	\]
\end{theorem}

The algorithm behind \pref{thm:potu} relies on iteratively running the random walk of Lovett and Meka \cite{LM15} to obtain a good partial coloring on the  rows that behave ``randomly'' on the set of current active coordinates. When $\lambda(A)$ is small, one can make sure that the rows behave ``randomly'' as long as their $\ell_2$-mass on the active coordinates is large enough. However, the bound on the discrepancy of the rows after their $\ell_2$-mass has become small crucially uses the fact that the instance is $\{0,1\}$-valued.

Our contributions here are new discrepancy bounds that generalize both \pref{thm:potu} and Banaszczyk's bound \cite{B98, BDG19}. Moreover, they hold both in the Beck-Fiala setting \textit{and} in the \Komlos setting. In \pref{sec:corollary}, we will apply these results to deduce random versions of \Komlos conjecture. Next, we state the two theorems that we will prove in this section. For simplicity, we focus on the case of square matrices ($m=n$), although we naturally expect the techniques to generalize to the $m\ge n$ case.

\begin{theorem}[Bound for pseudorandom Koml\'os instances]
	\label{thm:prkomlos}
	Let $A\in\mathbb{R}^{n\times n}$ be such that each column has $\ell_2$-norm at most $1$. Then there is a deterministic, polynomial-time algorithm to find $x\in\{\pm 1\}^n$ such that
	\[
	    \lVert Ax\rVert_\infty = O(1+\sqrt{\lambda(A) \log n})\mper
    \]
\end{theorem}

\begin{theorem}[Bound for pseudorandom Beck-Fiala instances]
	\label{thm:prbf}
	Let $A\in\{0,\pm 1\}^{n\times n}$ be such that each column has at most $s$ nonzero entries. Then there is a deterministic, polynomial-time algorithm to find $x\in\{\pm 1\}^n$ such that
	\[
	    \lVert Ax\rVert_\infty = O(\sqrt{s}+\min(\sqrt{\lambda(A) \log n}, \lambda(A)))\mper
	\]
\end{theorem}

As immediate corollaries, \pref{thm:prkomlos} implies \Komlos conjecture when $\lambda(A)=O(\frac{1}{\log n})$ and \pref{thm:prbf} implies Beck-Fiala conjecture when $\lambda(A)=O(\sqrt{s}+\frac{s}{\log n})$. This strictly improves \pref{thm:potu} for Beck-Fiala instances in the regime $\lambda(A)\in[O(\log n),O(s)]$. Furthermore, to the best of the authors' knowledge, \pref{thm:prkomlos} is the first result for pseudorandom instances of \Komlos conjecture.

\begin{remark}
    The mere column-sparsity assumption in \pref{thm:prbf} does not suffice to ensure that $\lambda(A)\le s$.\footnote{For example, consider a vector $v$ with half $+1$ and half $-1$ entries, take the first row of $A$ to be $v$ and fill the other rows with zeros. $A$ has one nonzero entry per column but $\| Av\|_2=\sqrt{n}\| v\|_2$.} However, as we will see (\pref{rem:bana}), we can essentially replace $\lambda(A)$ by $\min(\lambda(A),s)$ in the analysis. In this sense, our algorithm also matches Banaszczyk's bound.
\end{remark}

\subsection{Proof strategy and notations}
\label{sec:strategy}

We start by giving some idea of how we will use the fact that $\lambda(A)$ is small in our discrepancy framework. The main insight of \cite{P20} is that at any point in time in a discrepancy walk, we can control the $\ell_2$-mass restricted to active coordinates of all rows simultaneously.

\begin{lemma}[Lemma 2.3 in \cite{P20}]
	\label{lem:expander}
	Let $A\in \R^{n\times n}$ and $F\subseteq [n]$ be of size $k$. Then, for any constant $D>0$, there exists a subset $S\subseteq [n]$ such that $|S|\le k/D^2$ and for any $i\notin S$,
	\[
		\sum_{j\in F} A_{ij}^2\le \frac{k}{n}\sum_{j=1}^n A_{ij}^2+ D\lambda(A)\mper
	\]
\end{lemma}

Intuitively, the term $\frac{k}{n}\sum_{j\le n} A_{ij}^2$ corresponds to the $\ell_2$-squared-mass we would expect the row $A_i$ to have if the set of active coordinates $F$ were picked at random. The parameter $\lambda(A)$ gives a bound on the deviation from this random behavior. In particular, the $\ell_2$-mass of a row essentially decreases as in the average case as long as it is $\Omega(\sqrt{\lambda(A)})$. We recall the proof of \pref{lem:expander} for completeness.

\begin{proof}
    Let us denote $B_{ij}\defeq A_{ij}^2$ for all $i,j\in [n]$. If $u\in\R^n$ is a vector orthogonal to $\bf{1}$, then by definition we have $\sum_i \langle B_i,u\rangle^2\le \lambda(A)^2\|u\|_2^2$. Consider $u\in\R^n$ such that $u_j\defeq 1-\frac{k}{n}$ if $j\in F$ and $u_j\defeq -\frac{k}{n}$ if $j\notin F$. Then $u$ is orthogonal to $\bf 1$ and $\|u\|_2^2\le k$. Hence,
    \[
        \sum_{i=1}^n \left(\sum_{j\in F} A_{ij}^2-\frac{k}{n}\sum_{j=1}^n A_{ij}^2\right)^2\le \lambda(A)^2 k\mper
    \]
    The result follows from a simple counting argument.
\end{proof}

\paragraph{Roadmap of the proof.} In order to prove \pref{thm:prkomlos} and \pref{thm:prbf}, we will track two different types of potential functions, depending on which of the main term or the error term in \pref{lem:expander} dominates.  \pref{sec:pseudorandom} will be devoted to bounding the discrepancy incurred in the regime where the row mass decreases as if the input were random. The analysis here will mirror our proof of Spencer's theorem. In \pref{sec:small}, we will consider the case where the error term dominates. There we leverage the fact that the row mass has become small. In this setting, we will use a potential function that was introduced in \cite[Appendix B]{LRR17} to recover Banaszczyk's bound with the multiplicative weights update method.

Before starting the proof, we introduce some useful concepts and notations. From now on, we fix a matrix $A\in \R^{n\times n}$ with column $\ell_2$-norm bounded by 1. With the context being clear, we will write $\lambda\defeq \lambda(A)$. Our algorithm will follow the structure of the meta-algorithm \pref{alg:iterative}, therefore we will use our usual notations: $x(t)$ for the coloring at time $t$, $F(t)$ for the active coordinates at time $t$, etc.

\begin{definition}
    For each row $i\in [n]$, we define 
    \[
        t_i\defeq \min\left\{t\ge 0 : \sum_{j\in F(t)} A_{ij}^2\le 8\lambda\right\}\mper
    \]
\end{definition}

    In words, $t_i$ represents the time at which the $i$-th row stops behaving as if it were random. We will write $P(t)\defeq \{i\in [n] : t\le t_i\}$ for the set of ``pseudorandom'' rows at time $t$. The following observation explains what we mean by ``pseudorandom'' -- we can pretend as if the freezing process decreases linearly the $\ell_2$-mass of the rows. It is an easy consequence of \pref{lem:expander}.

\begin{claim}
    \label{claim:easypseudo}
    Fix any time step $t\ge 0$. There exists a subset of rows $I=I(t)\subseteq [n]$ such that $|I|\le \frac{|F(t)|}{16}$ and for any $i\in P(t),i\notin I$:
    \[
        \sum_{j\in F(t)} A_{ij}^2\le \frac{2|F(t)|}{n}\sum_{j=1}^n A_{ij}^2\mper
    \]
\end{claim}

For any row $i\in [n]$, we will track separately the contributions to its discrepancy for $t\le t_i$ and $t>t_i$.

\paragraph{Random regime.} Similarly to \cite{P20}, we group together the rows that have similar \textit{total} $\ell_2$-mass. For any $r\in \{1,\ldots, \lceil\log_2 n\rceil\}$, let
\[
    R_r\defeq \left\{i\in [n] : \sum_{j=1}^n A_{ij}^2\in (2^{r-1},2^r]\right\}\mper
\]
We also consider $R_0\defeq \{i\in [n] : \sum_{j=1}^n A_{ij}^2\le 1$\}. An easy double counting argument bounds the size of each $R_r$: $|R_r|\le n 2^{1-r}$ for any $r\le \lceil \log_2 n\rceil$. Our strategy will be to play several Spencer's games in parallel (restricted to the rows in $R_r$, for each value of $r$) and carefully allocate some ``effective dimension'' to each of them at any step of the walk.

We now define $\pi_{r,t}:\R^n\to \R^n$ to be a projection to the coordinates of $R_r$ of the discrepancy of the rows that behave pseudorandomly. Once $t>t_i$, we keep tracking in the $i$-th row the same value $\langle A_i,x(t_i)\rangle$. In short, for any $x\in\R^n$ and $i\in [n]$,
\[
    (\pi_{r,t}(x))_i=\begin{cases}
        0\quad &\text{if $i\notin R_r$}\\
        \langle A_i,x(t_i)\rangle\quad &\text{if $i\in R_r$ but $i\notin P(t)$}\\
        \langle A_i,x\rangle \quad&\text{if $i\in R_r\cap P(t)$}
    \end{cases}
\]
Now we are able to define our potential functions in the random regime. For any $r=0,\ldots, \lceil \log_2 n\rceil$, let
\[
    \Phi_{r,t}(x)\defeq \omega^*_{\frac{1}{2},\sqrt{n2^{1-r}}}(\pi_{r,t}(x))\mper
\]
The choice of $\eta_r=\sqrt{n2^{1-r}}$ as regularization parameter can be justified by the fact that there are at most $n2^{1-r}$ rows in the $r$-th group -- so this is essentially the smallest value of $\eta_r$ that makes the additive approximation error of the regularized maximum $O(1)$ (which is our target discrepancy in this regime).

Our main lemma states that it is possible to design $\oracle$ with a bit of slack in the dimension requirements, in such a way that all the potential functions $\Phi_r$ only pay a constant amortized increase over the duration of the walk.

\begin{lemma}
    \label{lem:pseudo}
    There is a construction of $\oracle(A,x(t))$ that always returns a subspace of codimension at most $\frac{k}{4}+O(1)$ (with $k$ being the number of active coordinates of $x(t)$), such that for any $r=0,\ldots, \lceil \log_2 n\rceil$, $\Phi_{r,T}(x(T))\le O(1)$.
\end{lemma}

\paragraph{Small row regime.} Set $\eta\defeq \sqrt{\frac{\log n}{\lambda}}$ (it will be the parameter of the regularizer in this regime). We define $B\in\R^{n\times n}$ to be the following thresholded version of $A$: for any $i,j\in [n]$, let $B_{ij}\defeq A_{ij}$ if $A_{ij}^2\le \frac{1}{16K^2\eta^2}$ and $B_{ij}\defeq 0$ otherwise. We will later see in \pref{sec:together} that it is sufficient to monitor the discrepancy of $B$ instead of $A$.

Recall that at this point of the walk, each row will have effective $\ell_2$-mass $O(\lambda)$. For algorithms based on orthogonality constraints, there is not much difference between having bounded rows and bounded columns, so this justifies patching an algorithm for Banaszczyk's setting at this point. We implement the potential function from \cite{LRR17} in our regularization framework.

Let $\pi'_{t}:\R^n\to\R^n$ to be such that for any $x\in\R^n$ and $i\in [n]$,
\[
    (\pi'_{t}(x))_i=\begin{cases}
        \langle B_i,x-x(t_i)\rangle-K\eta\sum_{j=1}^n B_{ij}^2(x_j-x_j(t_i))^2\quad &\text{if $i\notin P(t)$}\\
        0 \quad&\text{if $i\in P(t)$}
    \end{cases}
\]
for some constant $K>0$ that we will fix in the proof. Finally, we define our potential function for this regime:
\[
    \Psi_t(x)\defeq \smax_\eta(\pi'_t(x))\mper
\]
Our main lemma states we can make this potential non-increasing (and moreover the subspace dimension necessary for that allows some slack).

\begin{lemma}
    \label{lem:small}
    There is a construction of $\oracle(A,x(t))$ that always returns a subspace of codimension at most $\frac{k}{4}+O(1)$ (where $k$ is the number of active coordinates of $x(t)$), such that $\Psi_T(x(T))\le \Psi_0(x(0))$.
\end{lemma}

\begin{remark}
    Some intuition for $\pi_t'$ comes from the fact that to get a coloring of discrepancy $\sqrt{\lambda\log n}$ for an input matrix with \textit{rows} of $\ell_2$-norm bounded by $\lambda$, the Chernoff and union bound argument suffices. If $v_1, \ldots, v_n$ are the rows of the matrix, it is essentially consisting in arguing that when $x\sim\{\pm 1\}^n$, it holds for any $\eta\ge 0$ that:
    \[
        \log \E \exp\left(\eta \max_{i\in [n]} |\langle v_i,x\rangle|\right)\le \log \E \sum_{i=1}^n \exp(\eta |\langle v_i,x\rangle|)\le \log \sum_{i=1}^n \exp\left(\frac{\eta^2\|v_i\|_2^2}{2}\right)\mper
    \]
    If we interpret $\|v_i\|_2^2$ as $\sum_{j\in [n]} v_{i,j}^2 x_j^2$, this might motivate us to look at the softmax of $\{\langle v_i,x\rangle-\eta \langle v_i^{\odot 2},x^{\odot 2}\rangle/2\}$.
\end{remark}

The plan for the rest of this section is as follows. First, we prove \pref{lem:pseudo} in \pref{sec:pseudorandom} and \pref{lem:small} in \pref{sec:small}. Then in \pref{sec:together} we show how to deduce \pref{thm:prkomlos} and \pref{thm:prbf}. Finally, we study the consequences of \pref{thm:prkomlos} and \pref{thm:prbf} for random instances in \pref{sec:corollary}.

\subsection{Discrepancy in the random regime}
\label{sec:pseudorandom}

Our main goal in this section is to prove \pref{lem:pseudo}. Throughout this discussion, we fix a small constant $\varepsilon\in (0,1/5)$. Our first lemma describes a construction of $\oracle$ that bounds locally the increase of the potential. This part is very similar in spirit to our proof of Spencer's theorem.

\begin{lemma}
    \label{lem:both}
    Let $x\defeq x(t)$ and $k=k(t)\defeq F(t)$. Let $R_0=\lceil\log_2(32n/k)\rceil$.  There exists a subspace $S=S(t)\subseteq F(t)$ such that $S$ has codimension at most $\frac{k}{4}$ and if $\delta\in S$ satisfies $\|\delta\|_\infty\le 1/\text{poly}(n)$,
    \[
        \Phi_{r,t}(x+\delta)-\Phi_{r,t}(x)\lesssim  u_r(\delta)+\frac{1}{k}\left(\frac{k2^r}{n}\right)^{\frac{1-3\varepsilon}{2}}\|\delta\|_2^2 \quad\text{ for any $r\le R_0$}
    \]
    and
    \[
        \Phi_{r,t}(x+\delta)=\Phi_{r,t}(x)\quad\text{for any $r>R_0$}
    \]
    where $\{u_r:r\le R_0\}$ are linear forms.
\end{lemma}

\begin{proof}
    For any $r=0,\ldots, R_0$, let $k_r=k_r(t)\defeq C_1 \left(\frac{k2^r}{n}\right)^{\varepsilon} k$ be the effective subspace dimension devoted to rows in the $r$-th group, where $C_1=C_1(\varepsilon)$ is chosen so that $\sum_{r\le R_0} k_r\le k/8$.

    Let $S'$ be the orthogonal complement of the span of the large rows and the row in $I$, namely $(\bigcup_{r>R_0} \{A_i:i\in R_r\})^\perp$. These rows all have total $\ell_2$-squared mass larger than $2^{R_0-1}\ge 16n/k$, so there are at most $k/16$ of them and thereby $S_1$ has codimension at most $k/16$.
    
    Applying \pref{lem:taylor} to $r\le R_0$, for some $\nabla_r\in \Delta_n$, it holds for any $\delta$ with $\|\delta\|_\infty\le 1/\text{poly}(n)$ that
    \[
        \Phi_{r,t}(x+\delta)-\Phi_{r,t}(x)\le u_r(\delta)+2\sqrt{n2^{1-r}}\sum_{i\in R_r\cap P(t)} \nabla_{r,i}^{\frac{3}{2}}\langle A_i,\delta\rangle^2\quad\text{for some linear form $u_r:\R^n\to\R$}\mper
    \]
    Let $I_r$ be the set of coordinates that are in the top $k_r/2$ entries of the gradient $\nabla_r$. We define $S_r$ to be the intersection of the orthogonal complement of the span of the rows in $I\cup I_r$ (where $I$ is the set of rows from \pref{claim:easypseudo} that satisfies $|I|\le k/16$), and of the top $k_r/2$-dimensional eigenspace of $\sum_{i\in R_r\cap P(t),i\notin I\cup I_r} \nabla_{r,i}^{\frac{3}{2}} A_i A_i^\top$ over $\R^{F(t)}$. Then if $\delta\in S_r$,
    \begin{align*}
        \Phi_{r,t}(x+\delta)-\Phi_{r,t}(x)-u_r(\delta)
        &\lesssim \frac{\|\delta\|_2^2\sqrt{n}}{2^{\frac{r}{2}} k_r} \sum_{i\in R_r\cap P(t),i\notin I\cup I_r} \nabla_{r,i}^{\frac{3}{2}}\sum_{j\in F(t)} A_{ij}^2&\text{(since $\delta\in S_r$)}\\
        &\lesssim \frac{\|\delta\|_2^2k}{k_r \sqrt{n} 2^{\frac{r}{2}}} \sum_{i\in R_r\cap P(t),i\notin I\cup I_r} \nabla_{r,i}^{\frac{3}{2}}\sum_{j=1}^n A_{ij}^2\quad&\text{(by \pref{claim:easypseudo} and $i\in P(t),i\notin I$)}\\
        &\le \frac{\|\delta\|_2^2 k 2^{\frac{r}{2}}}{k_r \sqrt{n}} \sum_{i\in R_r\cap P(t),i\notin I\cup I_r}  \nabla_{r,i}^{\frac{3}{2}} \quad&\text{(since $i\in R_r$)}\\
        &\lesssim \frac{\|\delta\|_2^2 k2^{\frac{r}{2}}}{k_r^{\frac{3}{2}} \sqrt{n}} \quad&\text{(since $i\notin I_r$)}\\
        &\lesssim \frac{1}{k} \left(\frac{k2^r}{n}\right)^{\frac{1-3\varepsilon}{2}} \|\delta\|_2^2\mper&\text{(by definition of $k_r$)}
    \end{align*}
    Finally we set $S\defeq S'\cap \,\bigcap_{r\le R_0} S_r$. One can check that $S$ has codimension at most $\frac{k}{16}+\frac{k}{16}+\sum_{r\le R_0} k_r\le \frac{k}{4}$.
\end{proof}

The following step is a trick to handle the first-order terms. Indeed, a caveat is that unlike in Spencer's setting, we cannot afford to move perpendicularly to all the gradients simultaneously. 

\begin{lemma}
    \label{lem:first}
    Fix $x\in\R^n$. Let $S$ be a subspace such that for any $\delta\in S$ and $r\le R_0$,
    \[
        \Phi_{r,t}(x+\delta)-\Phi_{r,t}(x)\lesssim  u_r(\delta)+\frac{1}{k}\left(\frac{k2^r}{n}\right)^{\frac{1-3\varepsilon}{2}}\|\delta\|_2^2\quad \text{for some linear forms $\{u_r:r\le R_0\}$}\mper
    \]
    Then for any $\delta\in S$, at least one of $+\delta$ or $-\delta$ satisfies that for any $r\le R_0$,
    \[
        \Phi_{r,t}(x\pm \delta)-\Phi_{r,t}(x)\lesssim  \frac{1}{k}\left(\frac{k2^r}{n}\right)^{\varepsilon}\|\delta\|_2^2\mper
    \]
\end{lemma}

\begin{proof}
    By picking $\varepsilon<1/5$, we have for any $\delta\in S$
    \[
        \sum_{r\le R_0} 2^{-\varepsilon r} (\Phi_{r,t}(x+\delta)-\Phi_{r,t}(x)-u_r(\delta))\le \frac{1}{k} \left(\frac{k}{n}\right)^\varepsilon \|\delta\|_2^2\mper
    \]
    By a trivial upper bound, this means that for any $\delta\in S,r\le R_0$,
    \[
        2^{-\varepsilon r} (\Phi_{r,t}(x+\delta)-\Phi_{r,t}(x)) \le \sum_{s\le R_0} 2^{-\varepsilon s} u_s(\delta)+\frac{1}{k} \left(\frac{k}{n}\right)^\varepsilon \|\delta\|_2^2\mper
    \]
    In particular, by picking the signing $\pm \delta$ that satisfies $\sum_{s\le R_0} 2^{-\varepsilon s} u_s(\pm \delta)\le 0$, we get that for any $r\le R_0$,
    \[
        \Phi_{r,t}(x\pm \delta)-\Phi_{r,t}(x)\lesssim  \frac{1}{k}\left(\frac{k2^r}{n}\right)^{\varepsilon}\|\delta\|_2^2\mper\qedhere
    \]
\end{proof}

\begin{proof}[Proof of {\pref{lem:pseudo}}]
    By combining \pref{lem:both} and \pref{lem:first}, we know that at any step where there are $k$ active coordinates remaining, the potential $\Phi_{r,t}$ increases by at most $\frac{1}{k}\left(\frac{k2^r}{n}\right)^{\varepsilon}\|\delta\|_2^2$ if $r\le 1+\log_2(32n/k)$ and is unchanged for $r\ge 1+\log_2(32n/k)$.
    
    Now fix any $r\le \lceil \log_2(n)\rceil$. By a similar argument to the one in the proof of \pref{thm:spencer}, after letting $\beta_k$ be the $\ell_2$-squared mass injected into $x$ starting from the first time for which there are at most $k$ active coordinates remaining, we can upper bound
    \[
        \Phi_{r,T}(x(T))-\Phi_{r,0}(0)\lesssim \left(\frac{2^r}{n}\right)^\varepsilon \sum_{k=1}^{64n2^{-r}} \frac{\beta_k-\beta_{k-1}}{k^{1-\varepsilon}}=O(1)\mper
    \]
    The claimed bound follows after noting that our choice of parameters for the regularizers also implies $\Phi_{r,0}(0)=O(1)$.
\end{proof}

\subsection{Discrepancy in the small row regime}
\label{sec:small}

Our next goal is to prove \pref{lem:small}, which is a direct consequence of the following lemma.

\begin{lemma}
    \label{lem:smalllocal}
    Fix any time $t$ and let $k\defeq |F(t)|$. There is a subspace $S$ of $\R^{F(t)}$ of codimension at most $\frac{k}{4}+O(1)$ such that
    \[
        \Psi_t(x+\delta)\le \Psi_t(x)
    \]
    holds for any $\delta\in S$ with $\|\delta\|_\infty\le 1/\text{poly}(n)$.
\end{lemma}

Before proving it, we recall the following well-known result (see for example \cite[Theorem 8]{BDG19} or \cite[Lemma 21]{LRR17}):

\begin{lemma}[]
	\label{lem:quadraticsub}
	For any $w_1,\ldots, w_m\ge 0$, $B\in \R^{m\times n}$, and $\alpha\in (0,1)$, there exists a subspace $S$ of $\R^n$ of codimension at most $\alpha n$ such that for all $\delta\in S$,
	\begin{align}
		\sum_{i\le m} w_i \left(\sum_{j\le n} B_{ij} \delta_j\right)^2\le \frac{1}{\alpha}\sum_{i\le m} w_i \sum_{j\le n} B_{ij}^2\delta_j^2\mper\label{eq:qsub}
	\end{align}
\end{lemma}

\begin{proof}
    Up to considering $\sqrt{w_i} B_i$, assume without of generality that $w_i=1$ for all $i\in [m]$. Moreover, the statement is invariant if we remove the zero columns from $B$ and replace $\delta_j$ by $\frac{\delta_j}{\|B^j\|_2}$ for all $j\in[n]$. Therefore we can also assume that all columns of $B$ have unit Euclidean length.
    
    Now the right-hand side is just $\frac{\|\delta\|_2^2}{\alpha}$ and the left-hand side is $\delta^\top \sum_{i\le m} B_i B_i^\top \delta$. We can simply choose $S$ to be the subspace of vectors $\delta$ orthogonal to the top $\alpha n$ eigenspace of the linear operator $\sum_{i\le m} B_i B_i^\top$, which has trace $n$. The result follows a counting argument.
\end{proof}

\begin{proof}[Proof of {\pref{lem:smalllocal}}]
    To avoid overcharging notations we drop in this proof the dependencies on $t$ and let $x\defeq x(t)$, $F\defeq F(t)$ and $P\defeq P(t)$. When $\|\delta\|_\infty\le 1/\text{poly}(n)$, we can apply Taylor expansion (\pref{lem:taylor}) -- for some $\nabla\in \Delta_n$:
    \begin{align*}
        \Psi_t(x+\delta)-\Psi_t(x)\le &\sum_{i\notin P} \nabla_i \sum_{j\in F} B_{ij} \delta_j+\eta \sum_{i\notin P} \nabla_i \left(\sum_{j\in F} B_{ij} \delta_j\right)^2\\
        &-2K\eta \sum_{i\notin P} \nabla_i \sum_{j\in F} B_{ij}^2 x_j \delta_j+2K^2\eta^3 \sum_{i\notin P^\top}\nabla_i \left(\sum_{j\in F} B_{ij}^2 x_j \delta_j\right)^2\\
        &-K\eta\sum_{i\notin P} \nabla_i \sum_{j\in F} B_{ij}^2 \delta_j^2+K^2\eta^3 \sum_{i\notin P} \nabla_i \left(\sum_{j\in F} B_{ij}^2 \delta_j^2\right)^2\mper
    \end{align*}
    First, note that the last term scales as $\delta_j^4$ so we can make it negligible by picking $\|\delta\|_\infty\le 1/\text{poly}(n)$.
    
    We consider the subspace $S_1$ of codimension at most 2 that is the orthogonal complement of the span of the 2 vectors from the linear terms in $\delta$ in the previous right-hand side. Also, by applying \pref{lem:quadraticsub} to two different matrices with $\alpha=\frac{1}{8}$, we can find $S_2$ of codimension $k/4$ such that any $\delta\in S_3$ satisfies the following two conditions:
	\begin{align}
	    \eta \sum_{i\notin P} \nabla_i \left(\sum_{j\in F}B_{ij}\delta_j\right)^2&\le 8\eta \sum_{i\notin P} \nabla_i\sum_{j\in F} B_{ij}^2 \delta_j^2.\label{eq:cond1small}\\
	    2K^2\eta^3 \sum_{i\notin P} \nabla_i \left(\sum_{j\in F} B_{ij}^2 x_j \delta_j\right)^2&\le 16K^2\eta^3\sum_{i\notin P} \nabla_i \sum_{j\in F} B_{ij}^4 x_j^2 \delta_j^2\label{eq:cond2small}\mper
	\end{align}
	Note that whenever \pref{eq:cond2small} is satisfied, it also follows from $|x_i|\le 1$ and the assumption $B_{ij}^2\le \frac{1}{16K^2\eta^2}$ that
	\[
    	2K^2\eta^3 \sum_{i\notin P} \nabla_i \left(\sum_{j\in F} B_{ij}^2 x_j \delta_j\right)^2\le \eta \sum_{i\notin P}  \nabla_i \sum_{j\in F} B_{ij}^2 \delta_j^2\mper
	\]
	% condition: 9 <= K 
    Let $S\defeq S_1\cap S_2$. $S$ has codimension at most $k/4+O(1)$ by construction. Picking $K\defeq 9$, we get
    \[
        \Psi_t(x+\delta)-\Psi_t(x)\le 0\mcom
    \]
    for any $\delta\in S$ satisfying $\|\delta\|_\infty\le 1/\text{poly}(n)$.
\end{proof}

Finally we bound the error of replacing $A$ by $B$.

\begin{lemma}
    \label{lem:other}
    For any $i\in [n]$,
    \[
        |\langle A_i-B_i, x(T)-x(t_i)\rangle|\lesssim \sqrt{\lambda\log n}\mper
    \]
\end{lemma}

\begin{proof}
    Fix $i\in [n]$. Let $F\defeq F(t_i)$ be the set of active coordinates when the $i$-th row becomes small. Since $\sum_{j\in F} A_{ij}^2=O(\lambda)$, it must be that
    \[
        |\{j\in F: A_{ij}-B_{ij}\neq 0\}|\lesssim \log n\mper
    \]
    Therefore, by Cauchy-Schwarz,
    \[
        |\langle A_i-B_i, x(T)-x(t_i)\rangle|\le \sum_{j\in F} |A_{ij}-B_{ij}|\lesssim \sqrt{\log n}\sqrt{\sum_{j\in F} (A_{ij}-B_{ij})^2}\lesssim \sqrt{\lambda\log n}\mper\qedhere
    \]
\end{proof}

\subsection{Putting everything together}
\label{sec:together}

We are now ready to prove \pref{thm:prkomlos} and \pref{thm:prbf}.

\begin{proof}[Proof of {\pref{thm:prkomlos}} and {\pref{thm:prbf}}]
    
We assume without loss of generality that $\max_i (Ax)_i=\|Ax\|_\infty$ by the usual trick of doubling the rows. We define $\oracle(A,x(t))$ to be the intersection of the halfspace from \pref{lem:both} and  \pref{lem:first}, and of the subspace from \pref{lem:smalllocal}.

On the one hand, \pref{lem:pseudo} implies that for any $r\le \lceil\log_2 n\rceil$ and $i\in R_r$,
\[
    \langle A_i,x(t_i)\rangle = \pi_{r,T}(x(T))_i\le \Phi_{r,T}(x(T))\le O(1)\mper
\]
On the other hand, \pref{lem:small} implies that for any $i\in [n]$,
\[
    \langle B_i,x(T)-x(t_i)\rangle-K\eta \sum_{j=1}^n B_{ij}^2 (x_j(T)-x_j(t_i))^2\le \Psi_T(x(T))\le \Psi_0(x(0))\le \sqrt{\lambda\log n}\mper
\]
Observe that 
\[
    \sum_{j=1}^n B_{ij}^2 (x_j(T)-x_j(t_i))^2\le 4\sum_{j=1}^n B_{ij}^2\lesssim\lambda\mper
\]
Hence, $\langle B_i,x(T)-x(t_i)\rangle\lesssim \sqrt{\lambda\log n}$, and by \pref{lem:other}, $|\langle A_i-B_i, x(T)-x(t_i)\rangle|\lesssim  \sqrt{\lambda\log n}$ holds as well. It remains to use the triangle inequality:
\[
    \|Ax(T)\|_\infty =O(1+\sqrt{\lambda \log n})\mper
\]
This implies \pref{thm:prkomlos} and the first part of \pref{thm:prbf}. For the second part of \pref{thm:prbf}, simply observe that if $A$ is a rescaled Beck-Fiala instance, namely $A_{ij}\in \{0,1/\sqrt{s}\}$ for all $i,j\in [n]$, then the condition $\sum_{j\in F(t_i)} A_{ij}^2= O(\lambda)$ implies by that ${A_i}$ has at most $O(s\lambda)$ nonzero entries in $F(t_i)$, and so $\sum_{j\in F(t_i)} |A_{ij}|=O(\lambda\sqrt{s})$. In particular, the time steps $t\in [t_i,T]$ affect the discrepancy of $A_i$ by at most $O(\lambda\sqrt{s})$. This shows that the constructed coloring also has discrepancy $O(1+\lambda\sqrt{s})$ in this case, which is equivalent to the second part of the bound in \pref{thm:prbf}.
\end{proof}

\begin{remark}
    \label{rem:bana}
    One may also recover Banaszszyk's bound for \Komlos (or Beck-Fiala) instances by repeating the argument from \pref{sec:small}, replacing $\lambda$ by some universal constant larger than 1 and adding an additional orthogonality constraint to large rows, so that the algorithm can pretend that the rows all have bounded $\ell_2$-mass. It follows from our previous observations on the link between the multiplicative weights update method and negative entropy regularization that this would be equivalent to the approach for recovering Banaszczyk's bound in \cite{LRR17}.
\end{remark}

%% file: consequences.tex
%!TEX root = main.tex

\subsection{Application to random instances}
\label{sec:corollary}

\subsubsection{Random orthogonal matrices}
\label{sec:random}

The next consequence of our bound for pseudorandom instances is that Koml\'os conjecture is true for \textit{random} rotation matrices. An equivalent geometric way to state our result is the following: there exists a universal constant $C>0$ such that when we randomly rotate the $n$-dimensional hypercube around the origin, with high probability there exists a corner at $\ell_{\infty}$-distance at most $C$ from the origin. 

Rotation matrices appear to be hard instances for proving the \Komlos conjecture, as present proof techniques merely manage to match the discrepancy bounds to those for the suprema of Rademacher processes involving the transpose matrix. Improving beyond $O(\sqrt{\log n})$ discrepancy for orthogonal matrices would therefore provide new techniques for treating Rademacher/Gaussian processes for structured matrices. A first step in making progress on this front would therefore be to consider \textit{random} orthogonal matrices.

What we mean by random rotation is a random matrix distributed according to the Haar measure on the orthogonal group $\mathcal O(n)$. The Haar measure is a natural generalization of the uniform distribution. We can just think of the sampling as picking the matrix columns to be i.i.d. standard Gaussians in $\R^n$ (which will be linearly independent almost surely), and orthonormalizing them with the Gram-Schmidt process. 

We explicitly computed small moments of the entries of such a random matrix with the help of the Maple package \texttt{IntHaar} \cite{GK21}.

\begin{claim}
	\label{claim:haar}
	Suppose $A$ is distributed according to the Haar measure on $\mathcal O(n)$. Then,
	\begin{align*}
		&\E \left[A_{11}^8\right]=\frac{105}{n(n+2)(n+4)(n+6)}\mcom\\
		&\E \left[A_{11}^4 A_{12}^4\right]=\frac{9}{n(n+2)(n+4)(n+6)}\mcom\\
		&\E \left[A_{11}^2 A_{12}^2 A_{21}^2 A_{22}^2\right]=\frac{n^2+4n+15}{n(n+2)(n-1)(n+1)(n+4)(n+6)}\mper
	\end{align*}
\end{claim}

\begin{corollary}[Koml\'os conjecture for random rotations]
	There is a deterministic algorithm that given a Haar-distributed random matrix $A$ on $\mathcal O(n)$, finds with high probability $x\in\{\pm 1\}^n$ such that
	\[
	    \| Ax\|_\infty = O(1)\mper
	\]
\end{corollary}

\begin{proof}
	Our proof is inspired by the observations in the proof of Theorem 1 of \cite{B01} for the Haar measure on the unitary group. Consider $B\defeq (A^{\odot 2})^\top A^{\odot 2}$. Using \pref{claim:haar}, we see that
	\begin{align*}
		\E \Tr B^2&=\sum_{1\le i,j,k,l\le n} \E\left[A_{ij}^2 A_{il}^2 A_{kj}^2 A_{kl}^2\right]\\
		&=n^2\E \left[A_{11}^8\right]+2n^2(n-1)\E \left[A_{11}^4 A_{12}^4\right]+n^2(n-1)^2\E \left[A_{11}^2 A_{12}^2 A_{21}^2 A_{22}^2\right]\\
		&=1+O\left(\frac{1}{n}\right)\mper
	\end{align*}
	Note that $1$ and $\lambda(A)^2$ are eigenvalues of the positive semidefinite matrix $B$, so $1+\lambda(A)^4\le \Tr B^2$, and $\E \lambda(A)^4\le O\left(\frac{1}{n}\right)$ by the previous estimate. By Markov's inequality, $\lambda(A)\le \frac{1}{\log n}$ with high probability, and we conclude by applying \pref{thm:prkomlos}.
\end{proof}

\subsubsection{Random Gaussian matrices}\label{sec:gaussian}

Next we show that for matrices with random Gaussian entries, the corresponding $\lambda$ parameter is small. 
Without loss of generality, we assume that the input matrix is square, as otherwise (in the regime $m\geq n$) we can add extra columns while only worsening $\lambda$. 
We assume that each entry is sampled i.i.d. from $\mathcal{N}(0,\sigma^2)$ with $\sigma = \frac{1}{\sqrt{n}}$, so that all column norms are tightly concentrated around $1$, i.e.
\[
    \Pr(1 -\varepsilon \leq\|   A^j   \|_2^2 \leq  1+\varepsilon) \geq 1 - 2 \exp(-\frac{n \varepsilon^2}{8})\mcom
\]
which follows from standard concentration bounds.

\begin{claim}\label{claim:gaussianbd}
Given a random Gaussian matrix  $A \in \R^{n \times n}$, where entries are i.i.d. Gaussians $\mathcal{N}(0,\sigma^2)$ with $\sigma = \frac{1}{\sqrt{n}}$, one has that
\[
    \max_{\langle u, \mathbf{1}\rangle=0} \frac{ \| A^{\odot 2} u \|_2}{\| u\|_2 } \leq \frac{1}{\sqrt{n}}
\]
with high probability.
\end{claim}
\begin{proof}
\input{defproof.tex}
\end{proof}

This shows that by applying \pref{thm:prkomlos} random Gaussian matrices have discrepancy $O(1)$.
\begin{corollary}
Given a random Gaussian matrix  $A \in \mathbb{R}^{n \times n}$, where entries are i.i.d. Gaussians $\mathcal{N}(0,\sigma^2)$, $\sigma = \frac{1}{\sqrt{n}}$, there exists a deterministic algorithm that finds a coloring $x \in \{\pm 1\}^n$ such that $\|Ax \|_\infty = O(1)$.
\end{corollary}

\section{Discussion on symmetric Beck-Fiala instances}
\label{sec:future}

An interesting special case of the Beck-Fiala conjecture is when the matrix $A$ is the adjacency matrix of some $s$-regular graph. It turns out that in this setting, a folklore argument based on \Lovasz Local Lemma implies that there exists a coloring with discrepancy $O(\sqrt{s\log s})$. Although there is an algorithm to construct such a coloring in polynomial time, it is not captured by the iterative framework we introduced in this paper. It is in our opinion a great open problem to unify those two lines of work.

To formulate the problem more precisely, we provide a new streamlined and self-contained analysis of the algorithm matching the bound based on \Lovasz Local Lemma. In particular, it highlights the differences with the sticky walk approach. Our inspiration is an argument of \cite{AIS19} for finding a satisfying assignment of bounded degree $k$-SAT instances. 

\begin{theorem}[Folklore]
	\label{thm:bflll}
	There is a randomized algorithm that, given $A\in\{0,1\}^{n\times n}$ with at most $s$ nonzero entries per row and at most $s$ nonzero entries per column, finds with high probability in polynomial time a coloring $x\in\{\pm 1\}^n$ such that $\lVert Ax\rVert_{\infty} =O (\sqrt{s\log s})$.
\end{theorem}

\begin{proof}
	We will call a row $\textit{bad}$ (w.r.t. an implicit full coloring) when its discrepancy is larger than $4\sqrt{s\log s}$. We consider the following algorithm. First, we generate a uniformly random coloring. Then we repeat $t$ times the operation of picking the bad row with smallest index (unless there is none, in which case we stop) and resampling all the variables appearing in it.

	Since each constraint contains at most $s$ variables and each variable appears in at most $s$ constraints, any constraint has nonempty intersection with at most $s^2$ other constraints. 
	
	Define $C_t$ to be the set of all ordered sequences of $t$ constraints that have nonzero probability to be picked in that order by the algorithm. The execution of the algorithm can be described as a rooted forest of $t$ vertices, each one of these corresponding to a constraint that is picked. When a constraint is picked, it can create at most $s^2$ children, each of which corresponding to a constraint of lower index that intersects it and became bad after the resampling.

	Therefore, we can encode an element $c\in C_t$ by giving $\{c_i : \forall j\in\{1,\ldots,i-1\}, c_j<c_i\}$, and a rooted forest on $t$ vertices, each (except the roots) with labels between $1$ and $s^2$. It follows from standard combinatorics that
	\[
	    |C_t|\le 2^n \binom{2t}{t} s^{2t}= 2^n (2s)^{2t}\mper
	\]

	Fix a sequence of resampled constraints $c\in C_t$ and a sequence of $t+1$ colorings $u_1, \ldots, u_{t+1} \in\{\pm 1\}^n$. $c$ and $u$ can correspond to a potential execution of the algorithm only if $u_i$ is $u_{i+1}$ where the $c_i$-th constraint of $u_i$ is bad. Applying Chernoff bounds, we see that there can be at most $2^s/s^8$ such $u_i$'s. It follows by induction that there are at most $2^{st}/s^{8t}$ possible sequences $u_1, \ldots, u_{t}$. On the other hand, for any fixed $u_1, \ldots, u_{t+1}, c_1, \ldots, c_t$, the probability that the algorithm follows exactly this sequence of constraints and colorings is at most $2^{-st}$. Hence, by a union bound, the probability that the final coloring that the sequence of constraints is $c_1, \ldots, c_t$ is at most $2^n s^{-8t}$.

    To conclude, we can apply another union bound to get that the probability that the final discrepancy is larger than $\sqrt{2s\log s}$ is at most $|C_t| 2^n s^{-8t}$, which is $n^{-\Omega(1)}$ for some $t=n^{O(1)}$.
\end{proof}

Instead of working with fractional colorings, here we walk directly in the space of full colorings. While the row-sparsity assumption is not really restrictive in the sticky walk framework (as we can always pick update vectors orthogonal to large rows), it seems crucial for arguments based on \Lovasz Local Lemma.

\paragraph{A concrete family of instances.} We now introduce a family of discrepancy instances for which 
the tools we use to analyze our iterative framework fail to provide interesting bounds. It is an interesting question whether a more refined analysis will yield an improved discrepancy bound.

We believe these examples are essential in benchmarking attempts at improving discrepancy bounds, so we  will consider them to be candidate hard instances.

\begin{definition}[Twisted Hypercubes]
	The graph on one vertex is the only twisted hypercube of dimension $0$. A \textit{twisted hypercube} of dimension $d$ is then obtained by taking two copies of the same twisted hypercube of dimension $d-1$, and adding a matching between both vertex sets.\footnote{A slightly different construction consists in adding a matching between two \textit{potentially distinct} twisted hypercubes of dimension $d-1$. This is for example the convention chosen in one previous use of the term ``twisted hypercube'' in the litterature \cite{DPPQWZ18}. We believe it does not make much difference in the discrepancy setting.}
\end{definition}

Twisted hypercubes of dimension $d$ have $n=2^d$ vertices, each of degree $O(\log n)$. \pref{thm:bflll} implies that they have colorings of discrepancy $O(\sqrt{\log n \log \log n})$, but \pref{thm:prbf} only gives a bound of $O(\log n)$ (note that this bound would also follow from \cite{BF81}). We believe it would be interesting to find a construction of colorings of twisted hypercubes of discrepancy $o(\log n)$ using the sticky walk approach.\footnote{One could also consider randomly twisted hypercubes, obtained by taking recursively uniformly random matchings. They form a family of structured random instances that are not captured by our pseudorandom bounds.}

\begin{remark}
    Although smoothing the twisted hypercube is not sufficient to apply our bounds on pseudorandom Beck-Fiala instances, we can at least transform it into the adjacency matrix of a (multi-)graph with constant spectral expansion, while only losing an $O(1)$-additive factor on the discrepancy of any coloring. Therefore, our question on the discrepancy of symmetric instances could be reduced to that of the discrepancy of symmetric \textit{expanding} instances.
\end{remark}

\section*{Acknowledgments}

LP's work on this project has received funding from the European Research Council (ERC) under the European
Union’s Horizon 2020 Programme for Research and Innovation (grant agreement No. 834861). AV acknowledges the support of the French Agence Nationale de la Recherche (ANR), under grant ANR-21-CE48-0016 (project COMCOPT).
We thank Lorenzo Orecchia for many helpful conversations on the topic of discrepancy, and the anonymous reviewers for their insightful comments and suggestions.

%% file: defproof.tex
%! TEX root = main.tex

Let $B\defeq A^{\odot 2}-\frac{\mathbf{1}\mathbf{1}^\top}{n}$. Observe that
\[
    \max_{\|u\|_2=1,\langle u,\mathbf{1}\rangle = 0} \|A^{\odot 2} u\|_2\le \max_{\|u\|_2=1}\|B u\|_2\mper
\]
Now, $B$ is a matrix with i.i.d. entries such that $\E B_{11}=0$, $\E B_{11}^2=O(1/n^2)$ and $\E B_{11}^4=O(1/n^4)$, so by a standard result from random matrix theory (see e.g. Theorem 2.3.8 of \cite{taobook}), it holds that $n\|B\|_{\text{op}}=O(\sqrt{n})$ with high probability. It readily follows that $\lambda(A)\le 1/\sqrt{n}$ with high probability.

%% file: ellipsoid.tex
\appendix

\section{An iterative algorithm for ellipsoid discrepancy}
\label{sec:ellipsoid}

In view of attacking harder discrepancy problems such as the Beck-Fiala and \Komlos conjectures, it is worth noting that a major obstacle in achieving the conjectured bounds lies in the fact that once many variables get frozen to $\pm 1$, we have a smaller degree of freedom in choosing our update without having seen a significant decrease in the norms of the matrix rows, when restricted to the unfrozen coordinates. 

An alternate strategy would be to use an amortized analysis that bounds how much each coordinate of the diagonal matrix that we employ to upper bound the Hessian of the potential function has contributed so far. The fact that the Hessian can change quite drastically throughout the execution of the algorithm poses some difficulties in realizing this approach. However, we can show that in the case where the Hessian essentially stays constant, we can obtain interesting bounds. 

\paragraph{Ellipsoid discrepancy.} To this extent, we study a simpler discrepancy problem that we call \textit{ellipsoid discrepancy}, and that was initially introduced by Banaszczyk \cite{B90}. Given some positive semidefinite matrix $B\in \R^{n\times n}$, we consider the norm $\|\cdot\|_B$ on $\R^n$ defined by
\[
    \| z\|_{B}\defeq \sqrt{\langle z,Bz\rangle}\mcom\quad\text{ for all $z\in\R^n$}\mper
\]
We are interested in the following ``Euclidean'' version of the \Komlos problem: given a matrix $Q\in\R^{n\times n}$ whose columns have $\ell_2$-norm at most 1, find $x\in\{\pm 1\}^n$ such that $\|Qx\|_B$ is small. Banaszczyk \cite{B90} proved that there always exists a coloring of $Q$ achieving ellipsoid discrepancy $\sqrt{\Tr B}$ (and this is tight for any $B$, as can be seen by taking the columns of $Q$ to be an orthonormal basis of eigenvectors of $B$). It follows implicitly from other works, including the Gram-Schmidt walk algorithm \cite{BDGL19}, that this bound can be matched algorithmically (at least up to the leading constant).\footnote{We thank the anonymous reviewers for pointing this out.} We give a different algorithmic proof, based on our iterative meta-algorithm \pref{alg:iterative}, that highlights how measuring the discrepancy in an amortized sense is sometimes necessary.

\begin{theorem}\label{thm:ellipsoid}
	There is a deterministic algorithm running in polynomial time that given a positive semidefinite matrix $B\in\R^{n\times n}$ and a matrix $Q\in\R^{n\times n}$ with column $\ell_2$-norm bounded by $1$, returns $x\in\{\pm 1\}^n$ such that
	\[
	    \lVert Qx\rVert_B=O\left(\sqrt{\Tr B}\right)\mper
	\]
\end{theorem}

Our proof of \pref{thm:ellipsoid} uses the iterative machinery described in \pref{sec:generic}. Since the $\|\cdot\|_B$-norm squared is already a smooth degree-2 polynomial function, we will not need any regularization. However, we stress that repeating our analysis for Spencer's theorem from \pref{sec:spencer} would only give an ellipsoid discrepancy bound of $\sqrt{\Tr B\log n}$.

\begin{proof}
    Up to rotating $Q$ (without affecting the norm of its columns), we assume without loss of generality that $B$ is diagonal, with diagonal elements $D_1\ge \ldots \ge D_n\ge 0$.

	We make use of the meta-iterative algorithm \pref{alg:iterative}. Set $L\defeq 1/2$ (since there is no regularization involved, any constant would work here). We show how to construct $\oracle(A,x(t))$ with $F(t)\defeq \{j\in [n]:x_j(t)\notin \{-1,1\}\}$ and $k=k(t)\defeq |F(t)|$. For any $\delta\in\R^n$,
	\begin{align}
	    \label{eq:ellipinc}
		\lVert Q(x(t)+\delta)\rVert_D^2-\lVert Qx(t)\rVert_D^2 = 2\sum_{i=1}^n D_{i} \langle Q_i, x(t)\rangle\langle Q_i,\delta\rangle + \sum_{i=1}^n D_i \langle Q_i,\delta\rangle^2\mper
	\end{align}
	Let $S$ be the subspace of $\R^{F(t)}$ of all the vectors that are orthogonal to $Q_1, \ldots, Q_{\lfloor k/2\rfloor-1}$ (restricted to $\R^{F(t)}$). Now we use that if $\delta\in S$,
	\[
	    \sum_{i=1}^n D_i \langle Q_i,\delta\rangle^2\le D_{\lfloor k/2\rfloor} \left\langle \sum_{i=1}^n Q_i Q_i^\top,\delta\delta^\top\right\rangle\mper
	\]
	Restricted to the coordinates in $F(t)$, the matrix $\sum_{i\in [n]} Q_i Q_i^\top$ has trace $\sum_{i\in [n], j\in F(t)} Q_{ij}^2\le k$ by assumption on the norm of the columns of $Q$. Therefore, by averaging, we can find a 2-dimensional subspace of $S$ for which the previous quadratic form is $O\left(\|\delta\|_2^2\right)$. Finally, we return the intersection of this subspace with the halfspace that makes the first-order term non-positive in \pref{eq:ellipinc}.

	Now we study the total ellipsoid discrepancy incurred over time steps $t=0,\ldots T$. Let $\beta_{k}$ be the $\ell_2$-squared norm injected into $x(t)$ between $t_k\defeq \min\{t\ge 0:|F(t)|\le k\}$ and $T\defeq \min\{t\ge 0:|F(t)|\le 3\}$. In particular, we have $\beta_{k}\le k$. We sum by parts the second-order increase of $\pref{eq:ellipinc}$ over the execution of the algorithm,
	\begin{align*}
		\sum_{i=2}^{\lfloor n/2\rfloor} D_{i}(\beta_{2i+1}-\beta_{2i-1})
		&= \sum_{i=3}^{\lfloor n/2\rfloor} \beta_{2i-1} (D_{i-1}-D_i)+D_{\lfloor n/2\rfloor} \beta_{n}\\
		&\le \sum_{i=3}^{\lfloor n/2\rfloor} (2i-1) (D_{i-1}-D_i)+2nD_{\lfloor n/2\rfloor}\\
		&\lesssim \Tr D\mper
	\end{align*}
	Since the first-order term of the increase is always non-positive and the final step of \pref{alg:iterative} only changes the (squared) ellipsoid discrepancy by at most $O(\Tr D)$, the final coloring $x^*$ satisfies $\| Ax^*\|_D^2\lesssim \Tr D$.
\end{proof}